\documentclass[11pt,letterpaper]{amsart}

\usepackage{mathrsfs, amssymb,amsmath,url,amsthm,color,comment}
\usepackage{hyperref}
\usepackage{thm-restate}
\usepackage{todonotes}
\usepackage{algorithm}
\usepackage{algpseudocode}
\usepackage{hyperref,multirow}
\hypersetup{
    colorlinks = true,
    linkcolor = blue
    }
\usepackage[nameinlink,capitalise]{cleveref}

\usepackage{graphicx}

\usepackage[letterpaper,left=1in,top=1in,right=1in,bottom=1in]{geometry}

\newcommand{\ignore}[1]{}

\newtheorem{theorem}{Theorem}[section]

\newtheorem{example}[theorem]{Example}
\newtheorem{claim}[theorem]{Claim}
\newtheorem{corollary}[theorem]{Corollary}
\newtheorem{lemma}[theorem]{Lemma}
\newtheorem{observation}[theorem]{Observation}
\newtheorem{definition}[theorem]{Definition}
\newtheorem{proposition}[theorem]{Proposition}

\newcommand\rel[1]{\mathbb{#1}}

\newcommand\minabs{\mathrel{\lhd\kern-2pt\lhd}}
\newcommand\tuple[1]{#1}

\newcommand\gG{\mathscr G}

\DeclareMathOperator{\proj}{proj}

\renewcommand\tuple[1]{\mathbf{#1}}

\DeclareMathOperator{\scope}{S}

\DeclareMathOperator\Csp{CSP}
\DeclareMathOperator{\Pol}{Pol}

\DeclareMathOperator{\Aut}{Aut}

\DeclareMathOperator{\V}{\mathcal{V}}
\DeclareMathOperator{\constraints}{\mathcal{C}}
\DeclareMathOperator{\instance}{\mathcal{I}}

\DeclareMathOperator{\injinstances}{Inj}

\newcommand{\relstr}[1]{\mathbb{#1}}

\newcommand{\sA}{\mathbb A}
\newcommand{\sB}{\mathbb B}
\newcommand{\sC}{\mathbb C}

\newcommand{\sF}{\mathbb F}

\newcommand{\sK}{\mathbb K}

\newcommand\sX{\mathbb X}

\newcommand{\fG}{\mathscr G}

\newcommand{\cJ}{\mathcal J}

\author[T.~Nagy]
{Tom\'a\v s Nagy}
	\address{Institut f\"{u}r Diskrete Mathematik und Geometrie, FG Algebra, TU Wien}
	\email{tomas.nagy@email.com}
    \urladdr{http://dmg.tuwien.ac.at/nagy/}

\author[M.~Pinsker]
{Michael Pinsker}
	\address{Institut f\"{u}r Diskrete Mathematik und Geometrie, FG Algebra, TU Wien}
	\email{marula@gmx.at}
    \urladdr{http://dmg.tuwien.ac.at/pinsker/}

\thanks{
This research was funded in whole or in part by the Austrian Science Fund (FWF) [P 32337, I 5948]. For the purpose of Open Access, the authors have applied a CC BY public copyright licence to any Author Accepted Manuscript (AAM) version arising from this submission. This research is also funded by the European Union (ERC, POCOCOP, 101071674). Views and opinions expressed are however those of the author(s) only and do not necessarily reflect those of the  European Union or the European Research Council Executive Agency. Neither the European Union nor the granting authority can be held responsible for them.}

\begin{document}

\title[Strict width for CSPs over homogeneous strucures of finite duality]{Strict width for Constraint Satisfaction Problems over homogeneous strucures of finite duality}

\begin{abstract}
    We investigate the `local consistency implies global consistency' principle of \emph{strict width} among  structures within the scope of the Bodirsky-Pinsker dichotomy conjecture for infinite-domain Constraint Satisfaction Problems (CSPs). Our main result  implies  that for certain CSP templates within the scope of that conjecture, having bounded strict width has a concrete consequence on the expressive power of the template called \emph{implicational simplicity}. This in turn yields an explicit  bound on the \emph{relational width} of the CSP, i.e., the amount of local consistency needed to ensure the satisfiability  of any instance. Our  result applies to first-order expansions of any  homogeneous $k$-uniform hypergraph, but more generally to any  CSP template  under the assumption of finite duality and  general abstract conditions mainly on its  automorphism group. In particular, it overcomes the restriction to binary signatures in the pioneering  work~\cite{Wrona:2020b}.
\end{abstract}

\maketitle

\section{Introduction}

\subsection{CSPs, polymorphisms, and strict width}
Fixed-template \emph{Constraint Satisfaction Problems (CSPs)} are computational problems parameterized by relational structures as follows: for every  relational structure $\relstr A=(A;R_1,\ldots,R_m)$, called a \emph{template} of the CSP,  one obtains a computational problem $\Csp(\relstr A)$ where  given variables $x_1,\ldots,x_n$ and a list of constraints of the form $R_i(x_{i_1},\ldots,x_{i_\ell})$, one has to decide whether or not this list of constraints is \emph{satisfiable}, i.e., whether the variables can be assigned values in $A$ in such a way that all constraints are true statements in $\relstr A$. Examples of problems arising in this way are the problem of solving linear equations over a finite or infinite field (basically by choosing the linear reduct of the field as a template), 3-SAT (where an appropriate  template would consist of all ternary relations on a Boolean  domain), graph  $n$-coloring (where a template is the clique $K_n$), a class of scheduling problems called \emph{temporal CSPs}~\cite{temporalCSP} (where the relations of templates are given by quantifier-free formulas using the order of the rational numbers), \emph{phylogeny  CSPs}~\cite{PhylogenyCSP} (where suitable templates can be provided using quantifier-free formulas using the homogeneous C-relation), Graph-SAT problems~\cite{Schaefer-Graphs} (where templates have relations given by quantifier-free formulas over the countable random graph) or more generally  Hypergraph-SAT problems~\cite{hypergraphs}, or the model-checking problems for sentences in the logic MMSNP of Feder and Vardi~\cite{MMSNP,MMSNP-journal} or the logic  GMSNP~\cite{GMSNP,Bienvenu:2014}  (for which suitable infinite-domain templates can be provided using non-trivial model-theoretic techniques). We refer to~\cite{BodirskyBook} as a source of a huge variety of examples, in particular of the common case where $\relstr A$ has to be chosen infinite.

In the important special case of  \emph{finite-domain CSPs} (i.e., $\Csp(\relstr A)$ where $\relstr A$ has a finite domain), the set of multivariate functions on $\relstr A$ leaving the relations of $\relstr A$ invariant turns out to capture a variety of fundamental properties of the  CSP. Such functions are called \emph{polymorphisms}, and the set of all polymorphisms of a template $\relstr A$ the \emph{polymorphism clone} $\Pol(\relstr A)$ of $\relstr A$. Among the properties reflected in $\Pol(\relstr A)$ are not only the computational time-complexity of the CSP (up to log-space reductions), but also for example its expressibility in first-order logic~\cite{FO1, FO_char1,FO_char2}, its solvability by certain algorithmic principles such as local consistency checking (it follows from~\cite{Jeavsons98}) (more precisely, solvability by a Datalog program, a property called \emph{bounded width}), or a strengthening of the latter named \emph{bounded strict width}~\cite{FederVardi} where establishing local consistency of an instance yields partial solutions which are also globally consistent, i.e., extend to a global solution.

Strikingly, many of the properties of a finite-domain  CSP alluded to are already captured by an even more abstract invariant than the polymorphism clone $\Pol(\rel A)$ of a template $\relstr A$, namely the \emph{identities} satisfied therein~\cite{BJK05}. The study of the complexity of  CSPs via identities has been coined the \emph{algebraic approach}, and indeed has allowed for the application of deep methods and results from universal algebra. Most importantly, the algebraic approach culminated in  the celebrated resolution, independently due to Bulatov~\cite{Bulatov:2017} and Zhuk~\cite{Zhuk:2017,Zhuk:2020}, of the dichotomy conjecture  of Feder and Vardi~\cite{FederVardi}  characterizing the class P among finite-domain CSPs: if P\ $\neq$\ NP, then the finite templates with a polynomial-time solvable CSP are precisely those possessing some  polymorphism $w$ of some arity $n\geq 2$ which satisfies the identities $w(x,\ldots,x,y)\approx\cdots\approx w(y,x,\ldots,x)$ for all $x,y$; any such polymorphism is called a \emph{weak near-unanimity (wnu) polymorphism}. Already before that, the finite-domain templates with a CSP solvable by local consistency checking had been identified as those having wnu polymorphisms of \emph{all} arities  $n\geq 3$~\cite{BartoKozikBoundedWidth,MarotiMcKenzie,Maltsev-Cond}. As another example, central to the present article,  those finite CSP templates with bounded strict width have been shown to be precisely those with a \emph{near-unanimity polymorphism}, which is a polymorphism $f$ satisfying $x\approx f(x,\ldots,x)\approx f(x,\ldots,x,y)\approx\cdots\approx f(y,x,\ldots,x)$ for all $x,y$~\cite{FederVardi}. 

Despite the fact that many natural CSPs can only be given by an infinite template (see the examples and references above), our understanding of such CSPs lags way behind what we know for finite domains.  This is not only due to  the increasing technical difficulties brought along by infinity, but also due to the failure of identities to capture fundamental properties as neatly. On the positive side, early on Bodirsky and Ne\v{s}est\v{r}il established the standard model-theoretic assumption of \emph{$\omega$-categoricity}  on the template $\relstr A$ (stating that $\relstr A$ is ``close to finite'' in that it has only finitely many $k$-tuples for every $k\geq 1$ up to the equivalence given by the action of its automorphism group on  $k$-tuples) 
as a sufficient, natural, and well-studied  condition for $\Pol(\rel A)$ to capture $\Csp(\rel A)$  similarly as in the finite case~\cite{BN}. Moreover, this result  was lifted  in~\cite{Topo-Birk} to the abstract topological-algebraic structure of $\Pol(\rel A)$, an intermediate step towards identities. However, there are also various negative results for the algebraic structure  of $\Pol(\rel A)$ alone (given precisely by the identities satisfied therein) within the class of $\omega$-categorical structures concerning time complexity~\cite{Hrushovski-monsters,SymmetriesNotEnough, BartoBodorKozikMottetPinsker} or solvability of a CSP by local consistency checking~\cite{TopoRelevant,TopoRelevant-TAMS,Rydval:2020}, the latter even within the important class of  temporal CSPs. It is therefore  remarkable that the aforementioned characterization of bounded strict width via near-unanimity polymorphisms carries over unconditionally and almost verbatim (using \emph{quasi near-unanimity polymorphisms}) to all $\omega$-categorical templates, by a result of Bodirsky and Dalmau~\cite{BodirskyDalmau}.

The notion of strict width of a CSP was introduced by Feder and Vardi in~\cite{FederVardi}. For $k\leq \ell$, a CSP template $\relstr A$ has \emph{relational width $(k,\ell)$} if  the \emph{$(k,\ell)$-minimality algorithm} on any instance of $\Csp(\relstr A)$ detects a contradiction whenever the instance is unsatisfiable; this   algorithm roughly keeps lists of partial solutions for all $\ell$-element subsets of the variables of the instance and establishes consistency of these lists via $k$-element subsets. \emph{Bounded width} simply means relational width $(k,\ell)$ for some $(k,\ell)$.   The template $\relstr A$ has \emph{strict width} $k$ if there exists $\ell>k$ such that after running  $(k,\ell)$-minimality without deriving a contradiction, any partial solution consistent with the remaining lists of partial solutions actually expands to a solution of the entire instance; \emph{bounded strict width} means strict width $k$ for some $k$. In the Artificial Intelligence literature, this property would be paraphrased as  `strong $k$-consistency implies global consistency' and has been studied in particular in temporal and spatial reasoning~\cite{dechter}; the property not only enjoys the mentioned algebraic characterization via quasi near-unanimity polymorphisms, but can equivalently be described by a structural property of the template called \emph{decomposability} which for finite domains is due to Baker and Pixley~\cite{BakerPixley}; see~\cite[Chapter~8]{BodirskyBook} for the $\omega$-categorical case. A natural example of a CSP with strict width $(2,3)$ is the digraph acyclicity problem, which can be modelled as the CSP of the template $(\mathbb Q,<)$: the $(2,3)$-minimality algorithm will basically compute the transitive closure of the constraints imposed on the variables of an instance, and if it does not find a cycle, then any local  solution it kept for triples of variables extends to a global one.

\subsection{Strict width within the Bodirsky-Pinsker conjecture}

The Bodirsky-Pinsker conjecture (which has formulations of evolving strength~\cite{BPP-projective-homomorphisms,BartoPinskerDichotomy,Topo,BKO+17,BKO+19}) identifies a tame subclass of the class of  $\omega$-categorical templates  where polynomial-time solvability of a CSP might, for all we know,  correspond to the satisfaction of various specific identities such as a slight modification of the wnu identities. The constraint relations of templates $\relstr A$ of this class are given by quantifier-free formulas using the relations of a ``ground structure'' $\relstr B$ which is \emph{finitely bounded} and \emph{homogeneous}; the former meaning that the finite  substructures of $\relstr B$ can be described as precisely those avoiding a fixed finite set of forbidden patterns, and the latter meaning that  any first-order property of any tuple in $\relstr B$ is completely determined by the relations that hold on it (equivalently, $\relstr B$ has \emph{quantifier elimination}, i.e., all first-order information in $\relstr B$ is already encoded in its language). Standard examples of such  ground structures $\relstr B$ are the order of the rational numbers, the random graph, random hypergraphs, or the random partial order. The CSPs of templates $\relstr A$ arising in this way are problems which basically ask for the existence of a finite linear order,  graph, hypergraph, partial order, or generally  a finite substructure of $\relstr B$,  on the variables of an instance  such that the constraints, which are by definition expressed in the language of $\sB$,  are satisfied. We refer to~\cite{infinitesheep} for a recent account of this kind of computational  problems.

Wrona in pioneering work~\cite{Wrona:2020a, Wrona:2020b} set out to investigate the notion of strict width within the scope of the Bodirsky-Pinsker  conjecture with the goal of determining the amount of consistency required to ensure the existence of a solution of any instance. However, presumably for technical reasons, his achievements  came at the price of   additional rather artificial  assumptions: $\relstr B$ was required to be what is called there a \emph{liberal binary core}, meaning that the relations of $\sB$ are all binary, $\sB$ has all binary relations which are first-order definable from $\sB$ among its relations, and $\sB$ embeds every relational structure $\sC$ of size $3,\dots,6$ unless some substructure of $\sC$ of size $2$ does not embed into $\sB$.

The goal of the present work is to eliminate these shortcomings, obtaining Wrona's results in particular for Hypergraph-SAT problems, i.e., CSPs where the relations of the template are first-order definable over the $k$-uniform hypergraph, for any $k\geq 3$. In particular, we overcome the restriction to binary relations. Our general theorem firstly  assumes that the finite bounds of the ground structure $\relstr B$ are closed under homomorphisms; we say that $\relstr B$ has  \emph{finite duality}. This is an important case within the scope of the Bodirsky-Pinsker conjecture, and  includes e.g.~the templates modelling Graph-SAT and Hypergraph-SAT problems, or the model-checking problems for MMSNP and GMSNP mentioned above. Secondly, we assume that $\sB$ be \emph{$k$-neoliberal}, meaning that all its relations are $k$-ary and in a strong sense the entire structure of $\relstr B$ is completely reflected in these relations; a bit more precisely, there are no non-trivial first-order definable relations of arity smaller than $k$, all first-order definable relations or arity $k$ are present in $\sB$, all first-order definable  relations of higher arity are Boolean combinations of those of arity $k$, and that there are no algebraic dependencies between the elements of any $k$-tuple.  The CSP literature encompasses  several examples of  templates   over $k$-neoliberal ground structures, e.g. many $k$-uniform hypergraphs for  $k\geq 3$~\cite{hypergraphs} or the homogeneous C-relation (with $k=3$)~\cite{PhylogenyCSP}.

Similarly as in~\cite{Wrona:2020b}, we prove that any CSP template $\sA$ which is a first-order expansion of a ground structure $\relstr B$ as above and which has bounded strict width  has limited expressive power in the form of \emph{implicational simplicity}. We will show that this condition is equivalent to saying that $\sA$ does not primitively positively define an injective relation which entails a formula of the form $R(x_1,\dots,x_k)\Rightarrow R(x_{i_1},\dots,x_{i_k})$ for some relation $R$ primitively positively definable from $\sA$. 

\begin{restatable}{theorem}{implsimple}\label{thm:implsimple}
Let $k\geq 3$, let $\sB$ be $k$-neoliberal, and suppose that $\sB$ has finite duality. Suppose that $\relstr A$ is a CSP template which is an expansion of $\sB$ by first-order definable relations (i.e., Boolean combinations of the relations of $\sB$). If $\relstr A$ has bounded strict width, then it is implicationally simple on injective instances.
\end{restatable}

As a corollary of~\Cref{thm:implsimple} and results from~\cite{SmoothApproximations, hypergraphs}, we obtain a bound on the amount of local consistency needed to solve CSPs of the templates under consideration.

\begin{restatable}{corollary}{kmain}\label{thm:kmain}
Let $k\geq 3$, let $\sB$ be $k$-neoliberal, and suppose that $\sB$ has finite duality. Suppose that $\relstr A$ is a CSP template which is an expansion of $\sB$ by first-order definable relations (i.e., Boolean combinations of the relations of $\sB$). If $\relstr A$ has bounded strict width, then it has relational width $(k,\max(k+1,b_{\sB}))$.
\end{restatable}

\subsection{Related work}

In~\cite{SmoothApproximations}, bounded width was characterized for CSP templates over several ground structures (in a binary signature, such as the random graph) using identities satisfied by \emph{canonical polymorphisms}; this amounts to assuming weaker identities than quasi near-unanimity identities under the assumption of their satisfaction by special polymorphisms (which cannot be used in this way over other ground structures such as the order of the rationals). The conditions given there were applied  in~\cite{SymmetriesEnough} to obtain a general  upper bound on the relational width of CSP templates satisfying them, and the bound was shown to be optimal for many templates. The bound on the relational width in~\Cref{thm:kmain} coincides with the bound proven in~\cite{SymmetriesEnough} for CSP templates which posses canonical \emph{pseudo-totally symmetric polymorphisms} of all arities $n\geq 3$. The results in~\cite{SymmetriesEnough} have been extended to CSP templates over finitely bounded homogeneous $k$-uniform hypergraphs for  $k\geq 3$ in~\cite{hypergraphs}.

Bounded strict width for $\omega$-categorical CSP templates was first studied in~\cite{BodirskyDalmau}  where the above-mentioned algebraic characterization was obtained. In~\cite{Wrona:2020a,Wrona:2020b}, an upper bound on the relational width of first-order expansions of certain binary structures with bounded strict width was given.

For finite-domain CSP templates, different algebraic conditions which are stronger than the ones characterizing bounded width and weaker than the ones characterizing bounded strict width have been  considered~\cite{Jonssonterms}. These conditions can also lifted to the $\omega$-categorical case similarly as in the case of near-unanimity polymorphisms -- it is therefore natural to ask if a similar bound as in~\Cref{thm:kmain} can be obtained  for templates satisfying them.

\section{Preliminaries}\label{sect:Prelims}

\subsection{Relational structures and permutation groups}

Let $\sB$ be a relational structure, and let $\phi$ be a first-order formula over the signature of $\sB$. We identify the interpretation $\phi^{\sB}$ of $\phi$ in $\sB$ with the set of satisfying assignments for $\phi^{\sB}$. Let $V$ be the set of free variables of $\phi$, and let $\tuple u$ be a tuple of elements of $V$. We define $\proj_{\tuple u}(\phi^{\sB}):=\{f(\tuple u)\mid f\in \phi^{\sB}\}$. 
A \emph{first-order expansion} of a structure $\sB$ is an expansion of  $\sB$ by relations which are first-order definable in $\sB$, i.e., of the form $\phi^\sB$.

A first-order formula  is called  \emph{primitive positive} (pp) if it is built exclusively from  atomic formulae, existential quantifiers, and conjunctions. A  relation is \emph{pp-definable} in a relational structure  if it is first-order definable by a pp-formula. 

\begin{definition}\label{defn:fb}
Let $\sB$ be a structure over a finite relational signature $\tau$. We say that $\sB$ is \emph{finitely bounded} if there exists a finite set $\mathcal{F}$ of finite $\tau$-structures such that for every finite $\tau$-structure $\sC$, $\sC$ embeds to $\sB$ if no $\sF\in\mathcal{F}$ embeds to $\sC$.
Let $\mathcal{F}_{\sB}$ be a set witnessing the finite boudnedness of $\sB$ such that the size of the biggest structure contained in $\mathcal{F}_{\sB}$ is the smallest possible among all choices of $\mathcal F$; we write $b_{\sB}$ for this size.

We say that $\sB$ has a \emph{finite duality} if it is finitely bounded with $\mathcal F$ being closed under homomorphic images.
\end{definition}

Let $\gG$ be a permutation group acting on a set $A$, let $k\geq 1$,  and let $\tuple a\in A^k$. The \emph{orbit} of $\tuple a$ under $\gG$ is the set $\{g(\tuple a)\mid g\in \gG\}$. We say that $\gG$ is \emph{oligomorphic} if for every $k\geq 1$, $\gG$ has only finitely many orbits of $k$-tuples in its action on $A$. We say that a relational structure $\sA$ is \emph{$\omega$-categorical} if its automorphism group is oligomorphic.
Let $k\geq 1$. We say that $\fG$ is \emph{$k$-transitive} if it has only one orbit in its action on injective $k$-tuple of elements of $A$. 
$\fG$ is \emph{$k$-homogeneous} if for every $\ell\geq k$, the orbit of every $\ell$-tuple under $\fG$ is uniquely determined by the orbits of its $k$-subtuples.
$\fG$ has \emph{no $k$-algebraicity} if the only fixed points of any stabilizer of $\fG$ by $k-1$ elements are these elements themselves, and we say that $\gG$ has \emph{no algebraicity} if it has no $k$-algebraicity, for every $k\geq 1$.
The \emph{canonical $k$-ary structure of $\fG$} is the relational structure on $A$ that has a relation for every orbit of $k$-tuples under $\fG$.

\begin{definition}
    Let $k\geq 2$, and let $\fG$ be a permutation group acting on a set $A$. 
    We say that $\fG$ is \emph{$k$-neoliberal} if it is oligomorphic, $(k-1)$-transitive, $k$-homogeneous, and has no $k$-algebraicity.

    A relational structure $\sB$ is \emph{$k$-neoliberal} if it is the canonical $k$-ary structure of a $k$-neoliberal permutation group.
\end{definition}

We remark that if a relational structure $\sB$ is $k$-neoliberal, then: by $(k-1)$-transitivity for every $\ell<k$, the only relations of arity $\ell$ which are first-order definable from $\sB$ are Boolean combinations of equalities and non-equalities; by definition, all $k$-ary first-order definable relations are the relations of $\sB$ itself; and it follows from  $k$-homogeneity and oligomorphicity that any relation of arity $\ell> k$ first-order definable from $\sB$ is a Boolean combination of the ($k$-ary) relations of $\sB$.

The notion of $k$-neoliberality is inspired by the notion of liberal binary cores introduced by Wrona in~\cite{Wrona:2020b} -- every liberal binary core is $2$-neoliberal. However, the opposite is not true -- the automorphism group of the universal homogeneous $\sK_3$-free graph (i.e., the unique homogeneous graph having the bound $\mathcal F=\{\sK_3\}$ in Definition~\ref{defn:fb}) is easily seen to be $2$-neoliberal, whence its expansion by the equality relation and by the relation containing all pairs of distinct  elements which are not connected by an edge is $2$-neoliberal, but it is a binary core which is not liberal. This is because a liberal binary core is supposed to be finitely bounded and the set of forbidden bounds should not contain any structure of size $k$ whenever $3\leq k\leq 6$. However, $\sK_3$ is a $3$-element graph which does not embed into $\sB$ but all its subgraphs of size at most $2$ do, and hence $\sK_3$ has to be contained in any set of forbidden bounds for the universal homogeneous $\sK_3$-free graph.

\begin{example}
For every $k\geq 2$, the automorphism group of the universal homogeneous $k$-uniform hypergraph is $k$-neoliberal.

Let $\sC_\omega^2$ be the countably infinite equivalence relation where every equivalence class contains precisely $2$ elements. Then $\Aut(\sC_\omega^2)$ is oligomorphic, $1$-transitive, and $2$-homogeneous, but not $2$-neoliberal. Indeed, for any element $a$ of $\sC_\omega^2$, the stabilizer of $\Aut(\sC_\omega^2)$ by $a$ fixes also the unique element of $\sC_\omega^2$ which is in the same equivalence class as $a$.

On the other hand, the automorphism group of the countably infinite equivalence relation with equivalence classes of  fixed size $m>2$ is easily seen to be $2$-neoliberal.
\end{example}

Note that if $\gG$ is a permutation group acting on a set $A$ which is $k$-neoliberal for some $k\geq 2$ and which is not equal to the group of all permutations on $A$, then the number $k$ is uniquely determined. Indeed, $k=\min \{i\geq 1\mid \gG\text{ is not }i\text{-transitive}\}$.

\subsection{Constraint satisfaction problems and bounded width}

For $k\geq 1$, we write $[k]$ for the set $\{1,\ldots,k\}$. 
Let $k,\ell\geq 1$, let $A$ be a non-empty set, let $i_1,\ldots,i_{\ell}\in[k]$, and let $R\subseteq A^k$ be a relation. We write $\proj_{(i_1,\ldots,i_{\ell})}(R)$ for the $\ell$-ary relation $\{(a_{i_1},\ldots,a_{i_{\ell}})\mid (a_1,\ldots,a_k)\in R\}$.
For a tuple $\tuple t\in A^k$, we write $\scope(\tuple t)$ for its \emph{scope}, i.e., for the set of all entries of $\tuple t$.
We write $I_k^A$ for the relation containing all injective $k$-tuples of elements of $A$. 

Let $\sA$ be a relational structure. An \emph{instance of $\Csp(\sA)$} is a pair $\mathcal \instance=(\V,\constraints)$, where $\V$ is a finite set of variables and $\mathcal C$ is a finite set of \emph{constraints}; for every constraint $C\in \mathcal C$, there exists a non-empty set $U\subseteq \V$ called the \emph{scope} of $C$ such that $C\subseteq A^U$, and such that $C$ can be viewed as a relation of $\sA$ by totally ordering $U$; i.e., there exists an enumeration $u_1,\ldots,u_k$ of the elements of $U$ and a $k$-ary relation $R$ of $\sA$ such that for all $f\colon U\to A$, it holds that  
$f\in C$ if, and only if, $(f(u_1),\dots,f(u_k))\in R$. The relational structure $\sA$ is called the~\emph{template} of the CSP.
A~\emph{solution} of a CSP instance~$\instance=(\V,\constraints)$ is a mapping $f\colon \V\rightarrow A$ such that for every $C\in\constraints$ with scope $U$, $f|_U\in C$.

An instance $\instance=(\V,\constraints)$ of $\Csp(\sA)$ is \emph{non-trivial} if it does not contain any empty constraint; otherwise, it is \emph{trivial}.
Given a constraint $C\subseteq A^U$ and a tuple $\tuple{v}\in U^k$ for some $k\geq 1$, the \emph{projection of $C$ onto $\tuple{v}$} is defined by $\proj_{\tuple{v}}(C):=\{f(\tuple{v})\colon f\in C\}$.
Let $U\subseteq \V$.

We denote by $\Csp_{\injinstances}(\sA)$ the restriction of $\Csp(\sA)$ to those instances where for every constraint $C$ and for every pair of distinct variables $u,v$ in its scope, $\proj_{(u,v)}(C)\subseteq I_2^A$.

\begin{definition}\label{def:minimality}
Let $1\leq k\leq \ell$. We say that an instance $\instance=(\V,\constraints)$ of $\Csp(\sA)$ is \emph{$(k,\ell)$-minimal} if both of the following hold:
\begin{itemize}
\item the scope of every tuple of elements of $\V$ of length at most $\ell$ is contained in the scope of some constraint in $\constraints$;
\item for every $m\leq k$, for every tuple $\tuple u\in \V^m$, and for all constraints $C_1, C_2 \in \constraints$ whose scopes contain the scope of $\tuple u$, the projections of $C_1$ and $C_2$ onto $\tuple u$ coincide.
\end{itemize}
We say that an instance $\instance$ is \emph{$k$-minimal} if it is $(k,k)$-minimal.
\end{definition}

Let $1\leq k$. If $\instance=(\V,\constraints)$ is a $k$-minimal instance and $\tuple u$ is a tuple of variables of length at most $k$, then there exists a constraint in $\constraints$ whose scope contains the scope of $\tuple u$, and all the constraints who do have the same projection onto $\tuple u$.
We write $\proj_{\tuple u}(\instance)$ for this projection, and call it the \emph{projection of $\instance$ onto $\tuple u$}.

Let $1\leq k\leq \ell$, let $\sA$ be an $\omega$-categorical relational structure, and let $p$ denote the maximum of $\ell$ and the maximal arity of the relations of $\sA$. 
Clearly not every instance $\instance=(\V,\mathcal{C})$ of $\Csp(\sA)$ is $(k,\ell)$-minimal.
However, every instance $\instance$ is \emph{equivalent} to an $(k,\ell)$-minimal instance $\instance'$ of $\Csp(\sA')$ where $\sA'$ is the expansion of $\sA$ by all at most $p$-ary relations pp-definable in $\sA$ in the sense that $\instance$ and $\instance'$ have the same solution set.
In particular we have that if $\instance'$ is trivial, then $\instance$ has no solutions. Moreover, $\Csp(\sA')$ has the same complexity as $\Csp(\sA)$ and the instance $\instance'$ can be computed from $\instance$ in polynomial time (see e.g., Section~2.3 in~\cite{SymmetriesEnough} for the description of the algorithm).

\begin{definition}
Let $1\leq k\leq \ell$. 
A relational structure $\sA$ has \emph{relational width $(k,\ell)$} if every non-trivial $(k,\ell)$-minimal instance equivalent to an instance of $\Csp(\sA)$ has a solution. $\sA$ has \emph{bounded width} if it has relational width $(k,\ell)$ for some $k,\ell$.
\end{definition}

If $\sA$ has relational width $(k,\ell)$, then we will also say that $\Csp(\sA)$ has relational width $(k,\ell)$. We say that $\Csp_{\injinstances}(\sA)$ has relational width $(k,\ell)$ if every non-trivial $(k,\ell)$-minimal instance equivalent to an instance of $\Csp_{\injinstances}(\sA)$ has a solution.

\begin{definition}
    Let $\sA$ be a relational structure, and let $k\geq 1$. We say that $\sA$ has \emph{strict width $k$} if there exists $\ell>k$ such that for every $(k,\ell)$-minimal instance $\instance=(\V,\constraints)$ of $\Csp(\sA)$, for every $U\subseteq \V$, and for every mapping $g\colon U\rightarrow A$ such that $g\in C|_U$ for every $C\in\constraints$, there exists a solution $f\colon \V\rightarrow A$ of $\instance$ such that $f|_U=g$. 
    We say that $\sA$ has \emph{bounded strict width} if it has strict width $k$ for some $k\geq 1$.
\end{definition}

\subsection{Polymorphisms}

Let $A$ be a set, let $k,n\geq 1$, and let $R\subseteq A^k$. A function $f\colon A^n\rightarrow A$ \emph{preserves} the relation $R$ if for all tuples $(a^1_1,\dots,a^1_k),\dots,(a^n_1,\dots,a^n_k)\in R$, it holds that the tuple $(f(a^1_1,\dots,a^n_1),\dots,f(a^1_k,\dots,a^n_k))$ is contained in $R$ as well. 
The function is a \emph{polymorphism} of a relational structure $\sA$ if it preserves all relations of $\sA$. The set of all polymorphisms of $\sA$ is denoted by $\Pol(\sA)$. The importance of polymorphisms is based on the fact that for $\omega$-categorical $\sA$, the pp-definable relations are precisely those preserved by all polymorphisms of $\sA$~\cite{BN}.

Strict width can be characterized algebraically  for CSPs over $\omega$-categorical templates $\sA$  by the existence of certain polymorphisms as follows. We say that $f\in\Pol(\sA)$ is \emph{oligopotent} if for every finite subset $B\subseteq A$, there exists $\alpha\in\Aut(\sA)$ such that $f(b,\ldots,b)=\alpha(b)$ for every $b\in B$. We say that $f$ is a \emph{quasi near-unanimity operation} if $f(a,\dots,a)=f(a,b,\dots,b)=f(b,a,b,\dots,b)=\dots=f(b,\dots,b,a)$ for every $a,b\in A$; $f$ is a local near-unanimity operation on a set $B\subseteq A$ if it satisfies $a=f(a,\dots,a)=f(a,b,\dots,b)=f(b,a,b,\dots,b)=\dots=f(b,\dots,b,a)$ for every $a,b\in B$.

\begin{theorem}[\cite{BodirskyDalmau}]\label{thm:characterization-strictwidth}
    Let $\sA$ be an $\omega$-categorical relational structure, and let $k\geq 2$. Then the following are equivalent.

    \begin{itemize}
        \item $\sA$ has strict width $k$,
        \item $\Pol(\sA)$ contains an oligopotent quasi near-unanimity operation of arity $k+1$,
        \item for every finite subset $B\subseteq A$, $\Pol(\sA)$ contains a local near-unanimity operation on $B$ of arity $k+1$.
    \end{itemize}
\end{theorem}

\section{Proof of the main result}

\subsection{Implicationally simple structures}

We introduce the notion of an implication and several related concepts  that will play a key role in the proof of~\Cref{thm:implsimple}. It is not hard to see that, unlike in the case for structures from~\cite{Wrona:2020b}, the reduction to $\Csp_{\injinstances}(\sA)$ is necessary since every structure in the scope of \Cref{thm:implsimple} is implicationally hard without restricting to injective instances.

\begin{definition}\label{def:implication}
Let $\sA$ be a relational structure. Let $V$ be a set of variables, let $\tuple u,\tuple v$ be injective tuples of variables in $V$ of length $k<|V|$ and $m<|V|$, respectively, such that $\scope(\tuple u)\cup \scope(\tuple v)=V$. Let $C\subseteq A^k$ and $D\subseteq A^m$ be pp-definable from $\sA$ and non-empty. We say that a pp-formula $\phi$ over the signature of $\sA$ with free variables from $V$ is a \emph{$(C,\tuple u,D,\tuple v)$-implication in $\sA$} if all of the following hold:
\begin{enumerate}
    \item for all distinct $x,y\in V$, $\proj_{(x,y)}(\phi^{\sA})\not\subseteq\{(a,a)\mid a\in A\}$,
    \item $C\subsetneq\proj_{\tuple u}(\phi^{\sA})$,
    \item $D\subsetneq\proj_{\tuple v}(\phi^{\sA})$,
    \item for every $f\in \phi^{\sA}$, it holds that $f(\tuple u)\in C$ implies $f(\tuple v)\in D$,
    \item for every $\tuple a\in D$, there exists $f\in\phi^{\sA}$ such that $f(\tuple u)\in C$ and $f(\tuple v)=\tuple a$.
\end{enumerate}

We say that $\phi$ is a \emph{$(C,\tuple u,D,\tuple v)$-pre-implication} if it satisfies items (2)-(5).
We will call $\phi$ a \emph{$(C,D)$-implication} if it is a $(C,\tuple u,D,\tuple v)$-implication for some $\tuple u\in I^V_k, \tuple v\in I^V_m$. We say that an implication $\phi$ is \emph{injective} if $\phi^{\sA}$ contains only injective mappings.

Let $\gG$ be a permutation group acting on $A$, and let $f\in\phi^{\sA}$. If $O,P$ are orbits under $\fG$ such that $f(\tuple u)\in O$, $f(\tuple v)\in P$, then we say that $f$ is an \emph{$OP$-mapping}.
\end{definition}

\begin{example}
Let $\sA$ be a relational structure, let $k\geq 1$, and let $O$ be an orbit of $k$-tuples under $\Aut(\sA)$. Suppose that $\sA$ pp-defines the equivalence of orbits of $k$-tuples under $\Aut(\sA)$. Then the formula defining this equivalence is an $(O,O)$-pre-implication in $\sA$. If $\sA$ is such that $\Aut(\sA)$ does not have any fixed point in its action on $A$, this pre-implication is an implication. For all orbits $P,Q$ of $k$-tuples under $\Aut(\sA)$, $\phi^{\sA}$ contains an $PQ$-mapping if, and only if, $P=Q$.
\end{example}

\begin{definition}
Let $\sA$ be a relational structure, and let $k\geq 1$.

The \emph{$k$-ary implication graph of $\sA$}, to be denoted by $\mathcal{G}_{\sA}$, is a directed graph defined as follows.

\begin{itemize}
    \item The set of vertices is the set of pairs $(C_1,C)$ where $\emptyset\neq C\subsetneq C_1\subseteq A^k$ and $C,C_1$ are pp-definable from $\sA$.
    \item There is an arc from $(C_1,C)$ to $(D_1,D)$ if there exists a $(C,\tuple u,D,\tuple v)$-implication $\phi$ in $\sA$ such that $\proj_{\tuple u}(\phi^{\sA})=C_1$, $\proj_{\tuple v}(\phi^{\sA})=D_1$.
\end{itemize}

The \emph{$k$-ary injective implication graph of $\sA$}, denoted by $\mathcal{G}_{\sA}^{\injinstances}$, is the (non-induced) subgraph of $\mathcal{G}_{\sA}$ which contains precisely the vertices $(C_1,C)$ where $C$ is injective and which contains an arc from $(C_1,C)$ to $(D_1,D)$ if $(C_1,C)\neq (D_1,D)$ and if there exists an injective $(C,\tuple u,D,\tuple v)$-implication $\phi$ in $\sA$ with $\proj_{\tuple u}(\phi^{\sA})=C_1$, $\proj_{\tuple v}(\phi^{\sA})=D_1$.

We say that $\sA$ is \emph{implicationally simple (on injective instances)} if the (injective) implication graph $\mathcal{G}_{\sA}$ ($\mathcal{G}_{\sA}^{\injinstances}$) is acyclic. Otherwise, $\sA$ is \emph{implicationally hard (on injective instances)}.
\end{definition}

Note that by item (1) in~\Cref{def:implication}, the implication graph does not necessarily contain all loops -- e.g., the formula over variables $\{x_1,\ldots,x_{2k}\}$ defined by $\bigwedge\limits_{i\in[k]}(x_i=x_{i+k})$ is not an implication.

The following is essentially subsumed by Lemma 3.3 in~\cite{SymmetriesEnough} but we provide the reformulation to our setting for the convenience of the reader.

\begin{lemma}\label{lemma:core-bwidth}
Let $k\geq 2$, let $\fG$ be a permutation group, let $\sB$ be its canonical $k$-ary structure, and suppose that $\sB$ is finitely bounded. Let $\instance=(\V,\constraints)$ be a non-trivial, $(k,\max(k+1,b_{\sB}))$-minimal instance of $\Csp(\sB)$ such that for every $\tuple v\in\V^k$, $\proj_{\tuple v}(\instance)$ contains precisely one orbit under $\fG$. Then $\instance$ has a solution.
\end{lemma}

\begin{proof}
Let $\sim$ be a binary relation defined on $\V$ such that $u\sim v$ if, and only if, $\proj_{(u,v)}(\instance)\subseteq \{(b,b)\mid b\in B\}$. Since $k\geq 2$, $\sim$ is an equivalence relation.

Let $\tau$ be the signature of $\sB$, and let us define a $\tau$-structure $\sA$ on $\V/\sim$ as follows. Let $R\in\tau$; then $R$ is of arity $k$. We set $([v_1]_{\sim},\ldots,[v_k]_{\sim})\in R$ if $\proj_{(v_1,\ldots,v_k)}(\instance)= R^{\sB}$. Note that by our assumption, for every $(v_1,\ldots,v_k)\in\V^k$, there is precisely one relation of $\sA$ containing the tuple $([v_1]_{\sim},\ldots,[v_k]_{\sim})$.

Let us show that the definition of $\sim$ does not depend on the choice of the representatives $v_1,\ldots,v_k\in\V$. We will show that it does not depend on the choice of $v_1$, the rest can be shown similarly. Let $u_1\sim v_1$, and let $C\in\constraints$ be such that $u_1,v_1,\ldots,v_k$ are contained in its scope. Then $C|_{\{u_1,v_1\}}$ consists of constant maps and it follows that $\proj_{(u_1,v_2,\ldots,v_k)}(\instance)=\proj_{(u_1,v_2,\ldots,v_k)}(C)=\proj_{(v_1,\ldots,v_k)}(C)=\proj_{(v_1,\ldots,v_k)}(\instance)$.

We claim that $\sA$ embeds into $\sB$. Suppose for contradiction that this is not the case. Then there exists a bound $\sF\in\mathcal{F}_{\sB}$ of size $b$ with $b\leq b_{\sB}$ such that $\sF$ embeds into $\sA$. Let $[v_1]_{\sim},\ldots,[v_b]_{\sim}$ be all elements in the image of this embedding. Find a constraint $C\in\constraints$ such that $v_1,\ldots,v_b$ are contained in its scope. Since $C$ is nonempty, there exists $f\in C$. Since all relations in $\tau$ are of arity $k$, since $\instance$ is $k$-minimal and since for every $\tuple v\in \V^k$ such that all variables from $\tuple v$ are contained in the scope of $C$, $\proj_{\tuple v}(C)$ contains precisely one orbit under $\gG$, it follows that $\sF$ embeds into the structure that is induced by the image of $f$ in $\sB$ which is a contradiction.

It follows that there exists en embedding $e\colon \sA\hookrightarrow \sB$ and it is easy to see that $f\colon \V\rightarrow B$ defined by $f(v):=e([v]_{\sim})$ is a solution of $\instance$.
\end{proof}

\begin{proposition}\label{prop:implsimple}
Let $k\geq 3$, let $\fG$ be a $(k-1)$-transitive oligomorphic permutation group, let $\sB$ be its canonical $k$-ary structure, and suppose that $\sB$ is finitely bounded. Let $\sA$ be a first-order expansion of $\sB$ which is implicationally simple on injective instances and such that $\Pol(\sA)$ contains a binary injection. Then $\sA$ has relational width $(k,\max(k+1,b_{\sB}))$.
\end{proposition}

\begin{proof}
By~\Cref{cor:libcores-purelyinj}, it is enough to show that $\Csp_{\injinstances}(\sA)$ has relational width $(k,\max(k+1,b_{\sB}))$. To this end, let $\instance=(\V,\constraints)$ be a non-trivial $(k,\max(k+1,b_{\sB}))$-minimal instance of $\Csp_{\injinstances}(\sA)$; we will show that there exists a satisfying assignment for $\instance$.

For every $i\geq 0$, we define inductively a $(k,\max(k+1,b_{\sB}))$-minimal instance $\instance_i=(\V,\constraints_i)$ of $\Csp_{\injinstances}(\sA)$ with the same variable set as $\instance$ such that $\instance_0=\instance$ and such that for every $i\geq 1$, $\constraints_i$ contains for every constraint $C_{i-1}\in\constraints_{i-1}$ a constraint $C_i$ such that $C_i\subseteq C_{i-1}$, and such that moreover for some $\tuple v\in\V^k$, it holds that $\proj_{\tuple v}(\instance_i)\subsetneq \proj_{\tuple v}(\instance_{i-1})$, or for every $\tuple v\in\V^k$, it holds that $\proj_{\tuple v}(\instance_i)$ contains only one orbit under $\fG$.

Let $\instance_0:=\instance$. Let $i\geq 1$. We define $\mathcal{G}_i$ to be the graph that originates from $\mathcal{G}_{\sA}^{\injinstances}$ by removing all vertices that are not of the form $(\proj_{\tuple v}(\instance_{i-1}),F)$ for some injective $\tuple v\in \V^k$, and some $F\subseteq A^k$. \Cref{claim:projs} implies that $\instance_{i-1}$ is $k$-minimal, and hence $\proj_{\tuple v}(\instance_{i-1})$ is well-defined. If $\mathcal G_i$ does not contain any vertices, let $\instance_i:=\instance_{i-1}$. Suppose now that $\mathcal G_i$ contains at least one vertex. In this case, since $\mathcal{G}_{\sA}^{\injinstances}$ and hence also $\mathcal G_i$ is acyclic, we can find a sink $(\proj_{\tuple v_i}(\instance_{i-1}),F_i)$ in $\mathcal{G}_i$ for some injective $\tuple v_i\in\V^k$. Let us define for every $C_{i-1}\in\constraints_{i-1}$ containing all variables from $\tuple v_i$ in its scope $C_{i}:=\{f\in C_{i-1}\mid f(\tuple v_i)\in F_i\}$, and let $C_i:=C_{i-1}$ for every $C_{i-1}\in\constraints_{i-1}$ that does not contain all variables from $\tuple v_i$ in its scope. Note that in both cases, $C_i\subseteq C_{i-1}$.
Finally, we define $\constraints_{i}=\{C_{i}\mid C_{i-1}\in\constraints_{i-1}\}$.

\begin{claim}\label{claim:projs}
For every $i\geq 1$, $\instance_i$ is non-trivial and $(k,\max(k+1,b_{\sB}))$-minimal. Moreover, for every $\tuple v\in \V^k\backslash\{\tuple v_i\}$, $\proj_{\tuple v}(\instance_{i})=\proj_{\tuple v}(\instance_{i-1})$ and $\proj_{\tuple v_i}(\instance_{i})=F_i$.
\end{claim}

Let $i\geq 1$ and if $i>1$, suppose that the claim holds for $i-1$. Note that if $\instance_i=\instance_{i-1}$, then there is nothing to prove so we may suppose that this is not the case. Observe that for every $C_i\in\constraints_i$ containing all variables from $\tuple v_i$ in its scope, $\proj_{\tuple v_i}(C_i)=F_i$ by the definition of $C_i$.
We will now show that for every $\tuple v\in \V^k\backslash\{\tuple v_i\}$ and for every $C_i\in\constraints_i$ containing all variables from $\tuple v$ in its scope, $\proj_{\tuple v}(C_i)=\proj_{\tuple v}(C_{i-1})$. Observe that if $C_i$ does not contain all variables from $\tuple v_i$ in its scope, then the conclusion follows immediately since $C_i=C_{i-1}$; we can therefore assume that this is not the case.
Assume first that $\tuple v$ is not injective, let $m$ be the number of pairwise distinct entries of $\tuple v$, and let $\tuple u$ be an injective $m$-tuple containing all variables of $\tuple v$. Hence, $\proj_{\tuple u}(C_i)=I_m^A=\proj_{\tuple u}(C_{i-1})$ by the $(k-1)$-transitivity of $\gG$ and it follows that $\proj_{\tuple v}(C_i)=\proj_{\tuple v}(C_{i-1})$. Now assume that $\tuple v$ is injective  and, striving for a contradiction, that $\proj_{\tuple v}(C_i)\subsetneq \proj_{\tuple v}(C_{i-1})$. It follows that $(\proj_{\tuple v}(C_{i-1}),\proj_{\tuple v}(C_i))$ is a vertex in $\mathcal G_{\sA}^{\injinstances}$ and hence also in $\mathcal G_i$.
Let $\tuple w=(w_1,\ldots,w_\ell)\in\V^\ell$ be an enumeration of all variables of $\tuple v$ and $\tuple v_i$. It follows that the pp-formula $\phi(\tuple w)$ defining $\proj_{\tuple w}(C_{i-1})$ is an injective $(F_i,\tuple v_i,\proj_{\tuple v}(C_i), \tuple v )$-implication. Hence, there is an arc from $(\proj_{\tuple v_i}(\instance_{i-1}),F_i)$ to $(\proj_{\tuple v}(C_{i-1}),\proj_{\tuple v}(C_i))$ in $\mathcal{G}_i$ and in particular, $(\proj_{\tuple v_i}(\instance_{i-1}),F_i)$ is not a sink in $\mathcal G_i$, a contradiction.

Now, it is easy to see that $\instance_i$ is $(k,\max(k+1,b_{\sB}))$-minimal. Indeed, since $\instance$ is $(k,\max(k+1,b_{\sB}))$-minimal, every subset of $\V$ of size at most $\max(k+1,b_{\sB})$ is contained in the scope of some constraint of $\instance$ and by construction also of $\instance_i$. Moreover, by the previous paragraph, any two constraint of $\instance_i$ agree on all $k$-element subsets of $\V$ within their scopes.
\medskip

Since for every $i\geq 0$, if $\mathcal{G}_i$ is not empty, we remove at least one orbit of $k$-tuples under $\fG$ from some constraint. By the oligomorphicity of $\gG$, there exists $i_0\geq 0$ such that $\mathcal G_{i_0}$ is empty. 
We claim that for every injective $\tuple v\in \V^k$, $\proj_{\tuple v}(\instance_{i_0})$ contains precisely one orbit of $k$-tuples under $\fG$: If $\proj_{\tuple v}(\instance_{i_0})$ contained more than one orbit, then $(\proj_{\tuple v}(\instance_{i_0}),O)$ would be a vertex of $\mathcal{G}_i$ for an arbitrary orbit $O\subseteq\proj_{\tuple v}(\instance_{i_0})$; $O$ being a relation of $\sA$ since $\sA$ is a first-order expansion of $\sB$.

It follows that $\instance_{i_0}=(\V,\constraints_{i_0})$ is a non-trivial, $(k,\max(k+1,b_{\sB}))$-minimal instance of $\Csp_{\injinstances}(\sB)$ that satisfies the assumptions of~\Cref{lemma:core-bwidth}. Hence, there exists a satisfying assignment for $\instance_{i_0}$ and whence also for $\instance$.
\end{proof}

\subsection{Binary injections and bounded width}

Here, we restate some results about binary injections from~\cite{SmoothApproximations} that will enable us to use Lemma 14 from~\cite{hypergraphs} in order to prove~\Cref{cor:libcores-purelyinj}.

We will use the following results from~\cite{SmoothApproximations}. The orbit $O$ with the property stated in~\Cref{lemma:bininj1} is called \emph{free} in~\cite{SmoothApproximations}.

\begin{lemma}[Proposition 21 in~\cite{SmoothApproximations}]\label{lemma:bininj1}
Let $\sB$ be a homogeneous structure such that there exists an orbit $O$ of pairs under $\Aut(\sB)$ with the property that for all $a,b\in B$, there exists $c\in B$ such that $(a,c),(b,c)\in O$, and let $\sA$ be first-order reduct of $\sB$. If $\Pol(\sA)$ contains an essential function, then it contains a binary essential function.
\end{lemma}

\begin{lemma}[Proposition 24 in~\cite{SmoothApproximations}]\label{lemma:bininj2}
Let $\sA$ be a first-order reduct of an $\omega$-categorical structure $\sB$ such that $\Aut(\sB)$ is $1$-transitive and such that its canonical binary structure has finite duality. If $\Pol(\sA)$ contains a binary essential function preserving $I_2^B$, then it contains a binary injection.
\end{lemma}

\begin{proposition}\label{prop:libcores-purelyinj}
Let $\ell\geq k\geq 2$, let $\sA$ be an $\omega$-categorical relational structure, and suppose that $\Pol(\sA)$ contains a binary injection. Then $\sA$ has relational width $(k,\ell)$ if, and only if, $\Csp_{\injinstances}(\sA)$ has relational width $(k,\ell)$.
\end{proposition}

\begin{proof}
    If $\sA$ has relational width $(k,\ell)$, then so does $\Csp_{\injinstances}(\sA)$ since every instance of $\Csp_{\injinstances}(\sA)$ is an instance of $\Csp(\sA)$.

    Suppose now that $\Csp_{\injinstances}(\sA)$ has relational width $(k,\ell)$.
    Let $f$ be an injective binary polymorphism of $\sA$. It follows that for every $\tuple s\in A^2$, $\tuple t\in I^2_A$, it holds that $f(\tuple s,\tuple t), f(\tuple t,\tuple s)\in I^2_A$. Let $\instance=(\V,\constraints)$ be a $(k,\ell)$-minimal non-trivial instance equivalent to an instance of $\Csp(\sA)$; we will show that $\instance$ has a solution. To this end, let $\instance'$ be an instance obtained from $\instance$ by identifying all variables $u,v\in \V$ with $\proj_{(u,v)}(\instance)\subseteq \{(a,a)\mid a\in A\}$, and by adding a constraint $C_{(u,v)}:=\{f\in A^{\{u,v\}}\mid f(u)\neq f(v)\}$ for every $u,v\in \V$ with $\proj_{(u,v)}(\instance)\not\subseteq \{(a,a)\mid a\in A\}$. It follows that $\instance'$ is a non-trivial instance of $\Csp_{\injinstances}(\sA)$ and moreover, every solution of $\instance'$ translates into a solution of $\instance$. Lemma 17 in~\cite{hypergraphs} yields that the $(k,\ell)$-minimal instance $\cJ$ equivalent to $\instance'$ is non-trivial and since $\instance'$ is an instance of $\Csp_{\injinstances}(\sA)$ which has relational width $(k,\ell)$, $\cJ$ has a solution which translates into a solution of $\instance$ as desired.
\end{proof}

\Cref{lemma:bininj1,lemma:bininj2,prop:libcores-purelyinj} immediately yield the following corollary, which will enable us to reduce $\Csp(\sA)$ for any structure $\sA$ in the scope of~\Cref{thm:implsimple} to $\Csp_{\injinstances}(\sA)$.

\begin{corollary}\label{cor:libcores-purelyinj}
Let $k\geq 2$, let $\gG$ be a $2$-transitive oligomorphic permutation group, let $\sB$ be its canonical $k$-ary structure, and suppose that $\sB$ is finitely bounded. Let $\sA$ be a first-order expansion of $\sB$, and suppose that $\sA$ has bounded strict width. Then $\sA$ has relational width $(k,\max(k+1,b_{\sB}))$ if, and only if, $\Csp_{\injinstances}(\sA)$ has relational width $(k,\max(k+1,b_{\sB}))$.
\end{corollary}

\begin{proof}
Since $\sA$ has bounded strict width, it has in particular bounded width, and hence $\Pol(\sA)$ does not have a uniformly continuous homomorphism to an affine clone by~\cite{LaroseZadori,wonderland}. In particular, it does not have a uniformly continuous clone homomorphism to the clone of projections. It is easy to see and well-known that $\Pol(\sA)$ then contains an essential function. Since $\gG$ is $2$-transitive, $I_2^B$ is an orbit under $\gG$, and it clearly satisfies the condition from~\Cref{lemma:bininj1}. It follows that $\Pol(\sA)$ contains a binary essential operation, and since $\sA$ is a first-order expansion of $\sB$, every polymorphism of $\sA$ preserves $I_2^B$. \Cref{lemma:bininj2} yields that $\Pol(\sA)$ contains a binary injection. Now, the statement follows directly from~\Cref{prop:libcores-purelyinj}.
\end{proof}

\subsection{Some implications with no bounded strict width}

In this section, we first prove that if a structure $\sA$ pp-defines certain implications, then it does not have bounded strict width (\Cref{lemma:l+1=rel,lemma:l+2=rel}). This will enable us to prove that if $\sA$ has bounded strict width, and if a relation pp-definable in $\sA$ contains a tuple with certain properties, then this relation contains an injective tuple with the same properties (\Cref{cor:noequality}).

\begin{lemma}\label{lemma:l+1=rel}
Let $k\geq 3$, an let $\sB$ be $k$-neoliberal. Let $\sA$ be a first-order expansion of $\sB$, let $\ell\in\{2,\ldots,k\}$, let $T\subseteq I_{\ell}^B$, and let $\phi$ be an $(T,=)$-implication in $\sA$ with $\ell+1$ variables. Then $\sA$ does not have bounded strict width.
\end{lemma}

\begin{proof}
Enumerate the variables of $\phi$ by $x_1,\ldots,x_{\ell+1}$. Without loss of generality, $\tuple u=(x_1,\ldots,x_{\ell})$ and $\tuple v=(x_{\ell},x_{\ell+1})$ are such that $\phi$ is an $(T,\tuple u,=,\tuple v)$-implication in $\sA$. The set $\phi^{\sA}$ can then be viewed as an $(\ell+1)$-ary relation $R(x_1,\ldots,x_{\ell+1})$.

Using the $k$-neoliberality of $\sB$, we can find $a_1\ldots,a_\ell,b_1,\ldots,b_\ell\in B$ with $a_\ell\neq b_\ell$ such that all of the following hold:
\begin{itemize}
    \item $(a_1,\ldots,a_\ell)\in T$,
    \item $(a_1,\ldots,a_{\ell-1},b_\ell)\in T$,
    \item $(b_1,\ldots,b_{\ell-1},a_\ell,b_\ell)\in R$.
\end{itemize}
To see this, let $(a_1,\ldots,a_\ell)\in T$ be arbitrary.
The fact that the automorphism group of $\sB$ has no $k$-algebraicity implies that there exists $b_\ell\in B$ which is distinct from $a_\ell$ but which lies in the same orbit under the stabilizer of $\Aut(\sB)$ by $a_1,\ldots,a_{\ell-1}$. In particular, $(a_1,\ldots,a_\ell)$ and $(a_1,\ldots,a_{\ell-1},b_\ell)$ lie in the same orbit under $\Aut(\sB)$, and hence $(a_1,\ldots,a_{\ell-1},b_\ell)\in T$. Finally, since $\Aut(\sB)$ is $2$-transitive and $\proj_{(\ell,\ell+1)}(R)\not\subseteq \{(a,a)\mid a\in A\}$, we have $I_2^B\subseteq\proj_{(\ell,\ell+1)}(R)$, and hence we can find $b_1,\ldots,b_{\ell-1}$ such that $(b_1,\ldots,b_{\ell-1},a_\ell,b_\ell)\in R$.

Suppose for contradiction that $\sA$ has bounded strict width. Then by~\Cref{thm:characterization-strictwidth}, there exist $m\geq 3$ and an $m$-ary $f\in \Pol(\sA)$ which is a local near-unanimity function on $\{a_1\ldots,a_\ell,b_1,\ldots,b_\ell\}$. Striving for contradiction, we will show by induction that for every $n$ with $0\leq n\leq m$, it holds that $a_\ell=f(b_\ell,\dots,b_\ell,a_\ell,\dots,a_\ell)$, where $b_\ell$ appears $n$-times. Applying this for $m=n$, we get $a_\ell=f(b_\ell,\dots,b_\ell)=b_\ell$, a contradiction. For $n=0$, our statement follows directly from the fact that $f$ is a local near-unanimity function on $\{a_1\ldots,a_\ell,b_1,\ldots,b_\ell\}$. Let therefore $n\geq 1$ and let us assume that the statement holds for $n-1$.

Since $\phi$ is an $(T,=)$-implication, $(b_1,\ldots,b_{\ell-1},a_\ell)\notin T$, and $(a_1,\ldots,a_\ell,a_\ell), (a_1,\ldots,a_{\ell-1},b_\ell,b_\ell)\in R$. Put
\begin{align*}
    \tuple t_b := \begin{pmatrix}
       a_1 \\
       \ldots \\
       a_{\ell-1} \\
       b_\ell \\
       b_\ell
    \end{pmatrix},
    \tuple t_{ab} := \begin{pmatrix}
       b_1 \\
       \ldots \\
       b_{\ell-1} \\
       a_\ell \\
       b_\ell
    \end{pmatrix},
    \tuple t_a := \begin{pmatrix}
       a_1 \\
       \ldots \\
       a_{\ell-1} \\
       a_\ell \\
       a_\ell
    \end{pmatrix}.
\end{align*}
By the discussion above, $\tuple t_x \in R$ for every $x\in\{a,b,ab\}$. 
Since $f$ preserves $R$, it follows that $$f(\tuple t_b,\dots, \tuple t_b,\tuple t_{ab},\tuple t_a,\ldots,\tuple t_a)\in R,$$ where $\tuple t_b$ appears exactly $(n-1)$-many times. This means that 
\begin{align*}
    \begin{pmatrix}
       a_1 &=& f(a_1,\dots,a_1,b_1,a_1,\ldots,a_1) \\
       \ldots && \ldots \\
       a_{\ell-1} &=& f(a_{\ell-1},\dots,a_{\ell-1},b_{\ell-1},a_{\ell-1},\ldots,a_{\ell-1}) \\
       a_\ell &=& f(b_\ell,\dots,b_\ell,a_\ell,a_\ell,\ldots,a_\ell) \\
       &&f(b_\ell,\dots,b_\ell,b_\ell,a_\ell,\ldots,a_\ell)
    \end{pmatrix} \in R.
\end{align*}

In the penultimate row, we use the inductive assumption since $b_\ell$ appears exactly $(n-1)$-times as an argument in $f$. In the last row, $b_\ell$ appears exactly $n$-many times. Since $(a_1,\ldots,a_\ell)\in T$ and $\phi$ is an $(T,=)$-implication, we get $f(b_\ell,\dots, b_\ell,a_\ell,\ldots,a_\ell)=a_\ell$, where $b_\ell$ appears $n$-many times, as desired. 
\end{proof}

\begin{lemma}\label{lemma:l+2=rel}
Let $k\geq 3$, and let $\sB$ be $k$-neoliberal. Let $\sA$ be a first-order expansion of $\sB$, let $\ell\in\{2,\ldots,k\}$, let $T\subseteq I_{\ell}^B$, and let $\phi$ be an $(T,=)$-implication in $\sA$ with $\ell+2$ variables. Then $\sA$ does not have bounded strict width.
\end{lemma}

\begin{proof}
Let us enumerate the variables of $\phi$ by $x_1,\ldots,x_{\ell+2}$. Without loss of generality, $\tuple u=(x_1,\ldots,x_{\ell})$ and $\tuple v=(x_{\ell+1},x_{\ell+2})$ are such that $\phi$ is an $(T,\tuple u,=,\tuple v)$-implication in $\sA$. Note that $S(\tuple u)\cap S(\tuple v)=\emptyset$ by the definition of an implication and since $\phi$ has $\ell+2$ variables. The set $\phi^{\sA}$ can then be viewed as an $(\ell+2)$-ary relation $R(x_1,\ldots,x_{\ell+2})$.
Moreover, we can assume that for every $i\in[\ell]$, and every $j\in\{\ell+1,\ell+2\}$, $\phi$ does not entail $T(\tuple u)\Rightarrow x_i=x_j$ in $\sA$;
otherwise, the result follows immediately from~\Cref{lemma:l+1=rel}. 

Suppose for contradiction that $\sA$ has bounded strict width. Using the $k$-neoliberality of $\sB$, we can find $a_1,\ldots,a_{\ell+1},b_1,\ldots,b_{\ell+1},c_1,d_2,\ldots,d_\ell\in B$ with $a_{\ell+1}\neq b_{\ell+1}$ such that all of the following hold:
\begin{itemize}
    \item $(a_1,\ldots ,a_\ell)\in T$,
    \item $(a_1,\ldots ,a_\ell,a_{\ell+1})\in I_{\ell+1}^B$,
    \item $(a_1,\ldots ,a_\ell,a_{\ell+1},a_{\ell+1})\in R$,
    \item $(c_1,a_2,\ldots,a_\ell)\in T$,
    \item $(c_1,a_2,\ldots,a_\ell,b_{\ell+1})\in I_{\ell+1}^B$,
    \item $(c_1,a_2,\ldots,a_\ell,b_{\ell+1},b_{\ell+1})\in R$,
    \item $(b_1,\ldots ,b_\ell,a_{\ell+1},b_{\ell+1})\in R$,
    \item $(b_1,d_2,\ldots,d_\ell)\in T$.
\end{itemize}
To find these elements, let first $(a_1,\ldots ,a_\ell)\in T$ be arbitrary. By our assumption above, $\phi$ does not entail $T(\tuple u)\Rightarrow x_i=x_j$ in $\sA$ for any $i\in[\ell]$, $j\in\{\ell+1,\ell+2\}$, and we can show that there exists $a_{\ell+1}\in B$ such that $(a_1,\ldots ,a_\ell,a_{\ell+1})\in I_{\ell+1}^B\cap\proj_{(1,\ldots,\ell+1)}(R)$. Indeed, for every $i\in[\ell]$, there exists $f_i\in \phi^{\sA}$ with $f_i(\tuple u)\in T$ and $f_i(x_i)\neq f_i(x_{\ell+1})$. By the same argument as in the proof of~\Cref{cor:libcores-purelyinj}, $\sA$ has a binary injective polymorphism $f$. Setting $g(x):=f(f\dots(f(f_1(x),f_2(x)),f_3(x))\dots,f_\ell(x))$, we get that $g\in\phi^{\sA}$, $g(\tuple u)\in T$, and $g(x_i)\neq g(x_{\ell+1})$ for every $i\in [\ell]$, and we can therefore set $a_{\ell+1}:=g(x_{\ell+1})$.
Since $\phi$ is an $(T,=)$-implication, it follows that $(a_1,\ldots ,a_\ell,a_{\ell+1},a_{\ell+1})\in R$.

By the fact that $\sB$ has no $k$-algebraicity, we can find $b_{\ell+1}$ such that $a_{\ell+1}\neq b_{\ell+1}$ and both these elements lie in the same orbit under the stabilizer of $\Aut(\sB)$ by $a_2,\ldots,a_{\ell}$. In particular, it follows that $(a_2,\ldots,a_\ell,b_{\ell+1})$ and $(a_2,\ldots,a_{\ell+1})$ lie in the same orbit under $\Aut(\sB)$, and hence there exists $c_1\in B$ such that $(c_1,a_2,\ldots,a_\ell,b_{\ell+1})$ and $(a_1,\ldots,a_{\ell+1})$ lie in the same orbit under $\Aut(\sB)$. In particular, $(c_1,a_2,\ldots,a_\ell)\in T$ and it follows that $(c_1,a_2,\ldots,a_\ell,b_{\ell+1},b_{\ell+1})\in R$. Since $(a_{\ell+1},b_{\ell+1})\in I^B_2\subseteq \proj_{(\ell+1,\ell+2)}(R)$, there exist $b_1,\ldots,b_\ell\in B$ such that $(b_1,\ldots ,b_\ell,a_{\ell+1},b_{\ell+1})\in R$. Finally, since $\Aut(\sB)$ is $1$-transitive, $\proj_{(1)}(T)=B$, and hence there exist $d_2,\ldots,d_\ell\in B$ such that  $(b_1,d_2,\ldots,d_\ell)\in T$.

Since $\sA$ has bounded strict width, \Cref{thm:characterization-strictwidth} yields that there exist $f\in \Pol(\sA)$ of arity $m\geq 3$ which is a local near-unanimity operation on $\{a_1,\ldots,a_{\ell+1},b_1,\ldots,b_{\ell+1},c_1,d_2,\ldots,d_\ell\}$. Similarly as in the proof of~\Cref{lemma:l+1=rel}, we will show by induction that for every $n$ with $0\leq n\leq m$, it holds that $a_{\ell+1}=f(b_{\ell+1},\dots,b_{\ell+1},a_{\ell+1},\dots,a_{\ell+1})$, where $b_{\ell+1}$ appears $n$-times. Applying this for $m=n$, we get $a_{\ell+1}=f(b_{\ell+1},\dots,b_{\ell+1})=b_{\ell+1}$, a contradiction. For $n=0$, our statement follows directly from the fact that $f$ is a local near-unanimity function on $\{a_1\ldots,a_{\ell+1},b_1,\ldots,b_{\ell+1}\}$. Let therefore $n\geq 1$ and let us assume that the statement holds for $n-1$.
Put
\begin{align*}
    \tuple t_b:=\begin{pmatrix}
       c_1 \\
       a_2 \\
       \ldots \\
       a_\ell \\
       b_{\ell+1} \\
       b_{\ell+1}
    \end{pmatrix},
    \tuple t_{ab}:=\begin{pmatrix}
       b_1 \\
       b_2 \\
       \ldots \\
       b_\ell \\
       a_{\ell+1} \\
       b_{\ell+1}
    \end{pmatrix},
    \tuple t_a :=\begin{pmatrix}
       a_1 \\
       a_2 \\
       \ldots \\
       a_\ell \\
       a_{\ell+1} \\
       a_{\ell+1}
    \end{pmatrix}.
\end{align*}
By the discussion above, $\tuple t_x \in R$ for every $x\in\{a,b,ab\}$.
Since $f$ preserves $R$, it follows that $\tuple t:= f(\tuple t_b,\dots,\tuple t_b,\tuple t_{ab},\tuple t_a,\dots,\tuple t_a)\in R$, where $\tuple t_b$ appears exactly $(n-1)$-many times, i.e.,
\begin{align*}
    \tuple t=
    \begin{pmatrix}
       &&f(c_1,\dots,c_1,b_1,a_1,\ldots,a_1) \\
       a_2 &=& f(a_2,\dots,a_2,b_2,a_2,\ldots,a_2) \\
       \ldots && \ldots \\
       a_\ell &=& f(a_\ell,\dots,a_\ell,b_\ell,a_\ell,\ldots,a_\ell) \\
       a_{\ell+1} &=& f(b_{\ell+1},\dots,b_{\ell+1},a_{\ell+1},a_{\ell+1},\ldots,a_{\ell+1}) \\
       &&f(b_{\ell+1},\dots,b_{\ell+1},b_{\ell+1},a_{\ell+1},\ldots,a_{\ell+1})
    \end{pmatrix} \in R.
\end{align*}
Note that we are using the inductive assumption in the penultimate row. Put
\begin{align*}
    \tuple s_b:=\begin{pmatrix}
       c_1 \\
       a_2 \\
       \ldots \\
       a_\ell
    \end{pmatrix},
    \tuple s_{ab}:=\begin{pmatrix}
       b_1 \\
       d_2 \\
       \ldots \\
       d_\ell
    \end{pmatrix},
    \tuple s_a :=\begin{pmatrix}
       a_1 \\
       a_2 \\
       \ldots \\
       a_\ell
    \end{pmatrix}.
\end{align*}
By construction, $\tuple s_x \in T$ for every $x\in\{a,b,ab\}$.
Since $f\in\Pol(\sA)$ and $T$ is pp-definable from $\sA$, $f$ preserves $T$ and it follows that $\tuple s:= f(\tuple s_b,\dots,\tuple s_b,\tuple s_{ab},\tuple s_a,\dots,\tuple s_a)\in T$, where $\tuple s_b$ appears exactly $(n-1)$-many times, i.e.,
\begin{align*}
    \tuple s=
    \begin{pmatrix}
       &&f(c_1,\dots,c_1,b_1,a_1,\ldots,a_1) \\
       a_2 &=& f(a_2,\dots,a_2,d_2,a_2,\ldots,a_2) \\
       \ldots && \ldots \\
       a_\ell &=& f(a_\ell,\dots,a_\ell,d_\ell,a_\ell,\ldots,a_\ell)
    \end{pmatrix} \in T.
\end{align*}
Hence, $\tuple s$ is precisely the tuple containing the first $\ell$ entries of $\tuple t$. Since $\phi$ entails $T(\tuple u)\Rightarrow x_{\ell+1}=x_{\ell+2}$ in $\sA$, it follows that the last two entries of $\tuple t$ are equal, i.e., $$a_{\ell+1} = f(b_{\ell+1},\dots,b_{\ell+1},a_{\ell+1},a_{\ell+1},\ldots,a_{\ell+1}) = f(b_{\ell+1},\dots,b_{\ell+1},b_{\ell+1},a_{\ell+1},\ldots,a_{\ell+1}),$$ where $b_{\ell+1}$ appears exactly $n$-times on the right side.
\end{proof}

The following corollary follows from~\Cref{lemma:l+1=rel,lemma:l+2=rel}.

\begin{corollary}\label{cor:noequality}
Let $k\geq 3$, let $\sA$ be a first-order expansion of a $k$-neoliberal $\sB$, and suppose that $\sA$ has bounded strict width. Let $\phi$ be a pp-formula over the signature of $\sA$ with variables from a set $V$ such that for all distinct $x,y\in V$, $\proj_{(x,y)}(\phi^{\sA})\not\subseteq\{(a,a)\mid a\in A\}$, and let $g\in\phi^{\sA}$. Then there exists an injective $h\in\phi^{\sA}$ with the property that for every $r\geq 1$ and for every $\tuple v\in I^V_r$, if $g(\tuple v)$ is injective, then $g(\tuple v)$ and $h(\tuple v)$ belong to the same orbit under $\Aut(\sB)$.
\end{corollary}

\begin{proof}
    Let $W$ be the set of all tuples $\tuple v\in I^V_r$ with $r\leq k$ and such that $g(\tuple v)$ is injective; we denote the orbit of $g(\tuple v)$ under $\Aut(\sB)$ by $O_{\tuple v}$. Note that all these orbits are pp-definable from $\sA$ -- indeed, for $r=k$ this is the fact that $\sA$ is a first-order expansion of $\sB$, and for $r<k$ this follows by the $r$-transitivity of $\Aut(\sB)$. Hence, the formula $$\psi\equiv \phi \wedge\bigwedge\limits_{\tuple v\in W}O_{\tuple v}(\tuple v)$$
    is equivalent to a pp-formula over $\sA$. Moreover, it does not entail in $\sA$ any equality among any two of its variables: Otherwise, take a subset $W' \subseteq W$ which is maximal with respect to inclusion with the property that if we replace $W$ by $W'$ in the above definition, then the resulting formula $\psi'$ does not entail in $\sA$ any equality among any two of its variables. Then taking any $\tuple v\in W\backslash W'$, $\psi'$ entails $O_{\tuple v}(\tuple v)\Rightarrow x=y$ in $\sA$ for some distinct $x,y\in V$ and by existentially quantifying all variables of $\psi'$ except for the variables from the set $\scope(\tuple v)\cup\{x,y\}$, we obtain an $(O_{\tuple v},\tuple v,=,(x,y))$-implication in $\sA$, in contradiction with~\Cref{lemma:l+1=rel,lemma:l+2=rel}.

    Since $\Aut(\sB)$ is $(k-1)$-transitive and $k\geq 3$, it is in particular $2$-transitive, and hence the relation $I_2^B$ is pp-definable from $\sB$ and hence also from $\sA$. It follows that $I_m^B$, where $m$ is the number of variables of $\psi$, is pp-definable from $\sA$ and hence, so is $\psi\wedge I_m$. Moreover, $\psi\wedge I_m$ is non-empty since otherwise, we would obtain an $(I_2^B,=)$-implication in contradiction with~\Cref{lemma:l+1=rel,lemma:l+2=rel} as in the previous paragraph with $I_2(x,y)$ for every $(x,y)\in I^V_2$ in the role of $O_{\tuple v}(\tuple v)$.

    Finally, observe that any $h\in (\psi\wedge I_m)^{\sA}$ has the desired property by the $k$-homogeneity of $\Aut(\sB)$.
\end{proof}

\subsection{Critical relations}

We adapt the notion of a critical relation from~\cite{Wrona:2020b} to our situation and prove that no structure which satisfies the assumptions of~\Cref{thm:implsimple} can pp-define such relation.

\begin{definition}\label{def:kcritical}
Let $k\geq 2$, and let $\sA$ be a relational structure. Let $C,D\subseteq I_k^A$ be disjoint and pp-definable from $\sA$, let $V$ be a set of $k+1$ variables, and let $\tuple u,\tuple v\in I^V_k$ be such that $\scope(\tuple u)\cup\scope(\tuple v)=V$ and such that $u_1,v_1\notin \scope(\tuple u)\cap \scope(\tuple v)$.
We say that a pp-formula $\phi$ over the signature of $\sA$ with variables from $V$ is \emph{critical in $\sA$ over $(C,D,\tuple u,\tuple v)$} if all of the following hold:
\begin{itemize}
    \item $\phi$ is a $(C,\tuple u,C,\tuple v)$-implication in $\sA$,
    \item $D \subsetneq \proj_{\tuple u}(\phi^{\sA})\subseteq I_k^A$,
    \item $D \subsetneq \proj_{\tuple v}(\phi^{\sA})\subseteq I_k^A$,
    \item for every $\tuple a\in D$, there exists $f\in\phi^{\sA}$ such that $f(\tuple u)\in D$ and $f(\tuple v)=\tuple a$.
\end{itemize}
\end{definition}

\begin{lemma}\label{lemma:kcritical}
Let $k\geq 3$, and let $\sA$ be a first-order expansion of a $k$-neoliberal relational structure $\sB$. Suppose that there exists a pp-formula $\phi$ which is critical in $\sA$ over $(C,D,\tuple u,\tuple v)$ for some $k$-ary $C,D$, and some $\tuple u,\tuple v$. Then $\sA$ does not have bounded strict width.
\end{lemma}

\begin{proof}
First of all, observe that the formula $\phi_{\injinstances}:=\phi\wedge I_{k+1}$ is equivalent to a pp-formula over $\sA$ by the $2$-transitivity of $\Aut(\sB)$ and $\phi_{\injinstances}$ is still critical in $\sA$ over $(C,D,\tuple u,\tuple v)$ by~\Cref{cor:noequality}. Indeed, all items of~\Cref{def:kcritical} except for the first one depend only on $\proj_{\tuple u}(\phi^{\sA})$ and on $\proj_{\tuple v}(\phi^{\sA})$ and moreover, these projections are injective. Furthermore, for all distinct $x,y\in V$, $\proj_{(x,y)}(\phi^{\sA})\not\subseteq \{(a,a)\mid a\in A\}$, and hence~\Cref{cor:noequality} implies that for every $g\in\phi^{\sA}$, there exists $h\in\phi_{\injinstances}^{\sA}$ such that $g(\tuple u)$ and $h(\tuple u)$ belong to the same orbit under $\Aut(\sB)$, and so do $g(\tuple v)$ and $h(\tuple v)$. It follows that $\phi_{\injinstances}$ satisfies also the first item of~\Cref{def:kcritical}.

Let $V=\{x_1,\ldots,x_{k+1}\}$ be the set of variables of $\phi$. We can assume without loss of generality that $\tuple u=(x_1,\ldots,x_k)$ and $S(\tuple v)=\{x_2,\ldots,x_{k+1}\}$. 
Let $\tuple v=(x_{i_1},\ldots,x_{i_k})$. 
In the rest of the proof, for any $k$-tuple $(t_2,\ldots,t_{k+1})$, we write $\proj_{\tuple v}(t_2,\ldots,t_{k+1})$ for the tuple $(t_{i_1},\ldots,t_{i_k})$ by abuse of notation.

Suppose for contradiction that $\sA$ has bounded strict width -- hence, by~\Cref{thm:characterization-strictwidth}, there exists an oligopotent quasi near-unanimity operation $f\in\Pol(\sA)$ of arity $\ell\geq 3$. 
Let us define $\tuple w^2,\ldots,\tuple w^{k+1}\in A^\ell$ as follows. Let $d^2,\ldots,d^{k+1}\in A$ be arbitrary such that $\proj_{\tuple v}(d^2,\ldots,d^{k+1})\in D$, and let $\tuple w^j$ be constant with value $d^j$ for all $j\in\{2,\ldots,k+1\}$. 
Setting $J=[\ell]$ in~\Cref{claim:hlp} below, we get that $\proj_{\tuple v}(f(\tuple w^{2}),\ldots,f(\tuple w^{k+1}))\in C$. On the other hand, $\proj_{\tuple v}(f(\tuple w^{2}),\ldots,f(\tuple w^{k+1}))\in D$ since $D$ is pp-definable from $\sA$, and hence preserved by $f$, contradicting that $C$ and $D$ are disjoint.

\begin{claim}\label{claim:hlp}
For every $J\subseteq[\ell]$, the following holds. Let $\tuple w^{2},\ldots,\tuple w^{k+1}\in A^{\ell}$ be such that $\tuple w^{2},\ldots,\tuple w^{k}$ are constant tuples, $\proj_{\tuple v}(w_i^{2},\ldots,w_i^{k+1})\in C$ for all $i\in [\ell]\backslash J$, and $\proj_{\tuple v}(w_i^{2},\ldots,w_i^{k+1})\in D$ for all $i\in J$. Then $\proj_{\tuple v}(f(\tuple w^{2}),\ldots,f(\tuple w^{k+1}))\in C$.
\end{claim}

We will prove the claim by induction on $n:=|J|$. For $n=0$ the claim follows by the assumption that $C$ is pp-definable from $\sA$, and hence it is preserved by $f$.

Let now $n>0$, and suppose that~\Cref{claim:hlp} holds for $n-1$. The set $\phi_{\injinstances}^{\sA}$ can be viewed as a $(k+1)$-ary relation $R(x_1,\ldots,x_{k+1})$.
Let $m\in J$ be arbitrary, and set $J':=J\backslash \{m\}$. Using the $k$-neoliberality of $\Aut(\sB)$, we will find $\tuple w^1\in A^{\ell}$ such that all of the following hold:
\begin{itemize}
    \item $(w^1_i,\ldots,w^{k+1}_i)\in R$ for all $i\in [\ell]\backslash \{m\}$,
    \item $(w^1_i,\ldots,w^k_i)\in C$ for all $i\in [\ell]\backslash J$ and $(w^1_i,\ldots,w^k_i)\in D$ for all $i\in J'$,
    \item $(w^1_m,\ldots,w^k_m)\in C$ and $w^1_m\neq w^{k+1}_m$.
\end{itemize}
To find $\tuple w^1$, let first $i\in [\ell]\backslash J$. Since $\proj_{\tuple v}(w^{2}_i,\ldots,w^{k+1}_i)\in C$, it follows by the definition of a $(C,\tuple u,C,\tuple v)$-implication in $\sA$ that there exists $h\in \phi^{\sA}$ such that $h(\tuple v)=\proj_{\tuple v}(w^{2}_i,\ldots,w^{k+1}_i)$ and $h(\tuple u)\in C$. Set $w^1_i:=h(x_1)$. Let now $i\in J'$. Using the definition of a critical formula over $(C,D,\tuple u,\tuple v)$, we can find $h\in\phi^{\sA}$ such that $h(\tuple v)=\proj_{\tuple v}(w^{2}_i,\ldots,w^{k+1}_i)$ and $h(\tuple u)\in D$ similarly as above, and set $w^1_i:=h(x_1)$. 
It remains to find $w^1_m$ satisfying the last item. By the $(k-1)$-transitivity of $\Aut(\sB)$ and the fact that it has no $k$-algebraicity, we can extend any tuple $(a^2,\ldots,a^{k})\in I^A_{k-1}$ to injective tuples $(a^1,\ldots,a^k), (b^1,a^2,\ldots,a^k)\in O$ with $a^1\neq b^1$, for an arbitrary injective orbit $O$ of $k$-tuples under $\Aut(\sB)$. Applying this fact to the tuple $(w^{2}_m,\ldots,w^{k}_m)$ and to the orbit of $\proj_{\tuple v}(w^{2}_1,\ldots,w^{k+1}_1)$ under $\Aut(\sB)$, we get $w^1_m$ such that $(w^1_m,\ldots,w^k_m)\in C$ and $w^1_m\neq w^{k+1}_m$ as desired.

Note that $\tuple w^1,\ldots,\tuple w^{k}$ satisfy the assumptions of~\Cref{claim:hlp} for $J'$ in the role of $J$ up to permuting the order of the tuples. Indeed, $\tuple w^{2},\ldots,\tuple w^{k}$ are constant, and it holds that $(w^1_i,w^{2}_i,\ldots,w^{k}_i)\in C$ for all $i\in [\ell]\backslash J'$ and $(w^1_i,w^{2}_i,\ldots,w^{k}_i)\in D$ for all $i\in J'$. Since $|J'|=n-1$, the induction hypothesis yields that $(f(\tuple w^1),\ldots,f(\tuple w^{k}))\in C$.

Since $(w^1_m, w^{k+1}_m)\in I_2^A=\proj_{(1,k+1)}(R)$, there exist $a^2,\ldots,a^k \in A$ such that $$(w_m^1,a^{2},\ldots,a^{k},w^{k+1}_m)\in R.$$ For all $j\in\{2,\ldots,k\}$, let $\tuple q^j$ be the tuple obtained by replacing the $m$-th coordinate of $\tuple w^j$ by $a^j$. It follows that $(w^1_i,q^{2}_i,\ldots,q^{k}_i,w^{k+1}_i)\in R$ for all $i\in[\ell]$, and since $R$ is preserved by $f$, it holds that $(f(\tuple w^1),f(\tuple q^{2}),\ldots,f(\tuple q^{k}),f(\tuple w^{k+1}))\in R$.

Since $\tuple w^i$ is constant, and since $f$ is a quasi near-unanimity operation, we have $f(\tuple w^j)=f(\tuple q^j)$ for all $j\in\{2,\ldots,k\}$. This implies that $(f(\tuple w^1),f(\tuple q^{2}),\ldots,f(\tuple q^{k}))\in C$. Since $\phi$ entails $C(\tuple u)\Rightarrow C(\tuple v)$ in $\sA$, we get $\proj_{\tuple v}(f(\tuple q^{2}),\ldots,f(\tuple q^{k}),f(\tuple w^{k+1}))\in C$, and hence $\proj_{\tuple v}(f(\tuple w^{2}),\ldots,f(\tuple w^{k+1}))\in C$.
\end{proof}

\subsection{Composition of implications}
We introduce composition of implications which will play an important role in the rest of the article.

\begin{definition}\label{def:composition}
Let $\sA$ be a relational structure, let $k\geq 1$, let $C,D,E\subseteq A^{k}$ be non-empty, let $\phi_1$ be a $(C,\tuple u^1,D,\tuple v^1)$-implication in $\sA$, and let $\phi_2$ be a $(D,\tuple u^2,E,\tuple v^2)$-implication in $\sA$. Let us rename the variables of $\phi_2$ so that $\tuple v^1=\tuple u^2$ and so that $\phi_1$ and $\phi_2$ do not share any other variables.
We define $\phi_1\circ\phi_2$ to be the pp-formula arising from the formula $\phi_1\wedge \phi_2$ by existentially quantifying all variables that are not contained in $\scope(\tuple u^1)\cup \scope(\tuple v^2)$.

Let $\psi$ be a $(C,\tuple u,C,\tuple v)$-implication. For $n\geq 2$, we write $\psi^{\circ n}$ for the pp-formula $\psi\circ \cdots \circ \psi$ where $\psi$ appears exactly $n$ times.
\end{definition}

\begin{lemma}\label{lemma:kimpl}
Let $k\geq 3$, let $\sA$ be a first-order expansion of the canonical $k$-ary structure of a permutation group $\fG$. 
Let $\phi_1,\phi_2$ be as in~\Cref{def:composition}, and suppose that $\proj_{\tuple v^1}(\phi_1)= \proj_{\tuple u^2}(\phi_2)$.
Then $\phi:=\phi_1\circ \phi_2$ is a $(C,\tuple u^1,E,\tuple v^2)$-pre-implication in $\sA$. Moreover, for all orbits $O_1\subseteq \proj_{\tuple u^1}(\phi_1^{\sA})$, $O_3\subseteq \proj_{\tuple v^2}(\phi_2^{\sA})$ under $\gG$, $\phi^{\sA}$ contains an $O_1 O_3$-mapping if, and only if, there exists an orbit $O_2$ under $\gG$ such that $\phi_1^{\sA}$ contains an $O_1 O_2$-mapping and $\phi_2^{\sA}$ contains an $O_2 O_3$-mapping.

Suppose moreover that $\gG$ is $k$-neoliberal, that $\sA$ has bounded strict width, and that $\phi_1$ and $\phi_2$ are injective implications. Then $\phi$ is a $(C,\tuple u^1,E,\tuple v^2)$-implication in $\sA$. Restricting $\phi^{\sA}$ to injective mappings, one moreover obtains an injective $(C,\tuple u^1,E,\tuple v^2)$-implication, which for all injective orbits $O_1\subseteq C, O_3\subseteq E$ under $\fG$ contains an $O_1 O_3$-mapping if, and only if, $\phi^{\sA}$ contains such mapping.
\end{lemma}

\begin{proof}
    Let us assume as in~\Cref{def:composition} that $\tuple v^1=\tuple u^2$ and that $\phi_1,\phi_2$ do not share any further variables. Let $V_1$ be the set of variables of $\phi_1$, let $V_2$ be the set of variables of $\phi_2$, and let $V$ be the set of variables of $\phi$.
    We will first prove the last sentence of the first part of~\Cref{lemma:kimpl} about $O_1 O_3$-mappings.
    To this end, let $O_1\subseteq \proj_{\tuple u^1}(\phi_1^{\sA}),O_2\subseteq \proj_{\tuple v^1}(\phi_1^{\sA}),O_3\subseteq \proj_{\tuple v^2}(\phi_2^{\sA})$ be orbits under $\gG$, and suppose that $\phi_1^{\sA}$ contains an $O_1 O_2$-mapping $f$ and $\phi_2^{\sA}$ contains an $O_2 O_3$-mapping $g$. Using that $f|_{V_1\cap V_2}$ is contained in the same orbit under $\gG$ as $g|_{V_1\cap V_2}$, find a mapping $h\colon V_1\cup V_2\rightarrow A$ such that $h|_{V_1}$ is contained in the same orbit under $\gG$ as $f$ and $h|_{V_2}$ is contained in the same orbit as $g$. It follows that $h\in(\phi_1\wedge \phi_2)^{\sA}$, and hence $h|_{V}\in \phi^{\sA}$ is an $O_1 O_3$ mapping.
    On the other hand, if $\phi^{\sA}$ contains an $O_1 O_3$-mapping $h'$, we can extend it to a mapping $h\in (\phi_1\wedge \phi_2)^{\sA}$ such that $f:=h|_{V_1}\in \phi_1^{\sA}$ and $g:=h|_{V_2}\in \phi_2^{\sA}$ by the definition of $\phi$. Setting $O_2$ to be the orbit of $h(\tuple v^1)$, we get that $f\in \phi_1^{\sA}$ is an $O_1 O_2$-mapping and $g\in\phi_2^{\sA}$ is an $O_2 O_3$-mapping as desired.
    
    Observe now that the fact that $\phi$ is a $(C,\tuple u^1,E,\tuple v^2)$-pre-implication in $\sA$ follows from the previous paragraph. Indeed, it follows immediately that $\phi$ satisfies items (4) and (5) of~\Cref{def:implication}. To see that items (2) and (3) are satisfied as well, take any $g\in\phi_2^{\sA}$ with $g(\tuple v^2)\notin E$, let $O_2$ be the orbit of $g(\tuple u^2)$ under $\gG$, and let $O_3$ be the orbit of $g(\tuple v^2)$. Since $\proj_{\tuple v^1}(\phi_1)= \proj_{\tuple u^2}(\phi_2)$, we can find an $O_1 O_2$-mapping in $\phi_1^{\sA}$, and $\phi^{\sA}$ contains an $O_1 O_3$-mapping witnessing that $C\subsetneq \proj_{\tuple u^1}(\phi^{\sA}), C\subsetneq \proj_{\tuple v^2}(\phi^{\sA})$.

    To prove the second part of the lemma, we will prove that for all orbits $O_1\subseteq C, O_2\subseteq D, O_3\subseteq E$ under $\fG$ such that $\phi_1^{\sA}$ contains an $O_1 O_2$-mapping $f$ and $\phi_2^{\sA}$ contains an $O_2 O_3$-mapping $g$, $\phi^{\sA}$ contains an injective $O_1 O_3$-mapping $h$.
    Note that as in the previous paragraph, it is enough to find an injective mapping $h\in(\phi_1\wedge \phi_2)^{\sA}$ such that $h|_{V_1}$ is contained in the same orbit under $\gG$ as $f$ and $h|_{V_2}$ is contained in the same orbit as $g$. 
    Let $U_1$ be the set of all injective tuples of variables from $V_1$ of length at most $k$, and for every $\tuple v\in U_1$, let us denote the orbit of $f(\tuple v)$ by $O^f_{\tuple v}$. Similarly, let $U_2$ be the set of all injective tuples of variables from $V_2$ of length at most $k$, and for every tuple $\tuple v\in U_2$, let $O^g_{\tuple v}$ be the orbit of $g(\tuple v)$. 
    Let us define a formula $\psi$ with variables from $V_1\cup V_2$ by $$\psi\equiv\bigwedge\limits_{\tuple v\in U_1}O^{f}_{\tuple v}(\tuple v)\wedge \bigwedge\limits_{\tuple v\in U_2}O^{g}_{\tuple v}(\tuple v).$$ 
    Note that since $\sA$ is a first-order expansion of the canonical $k$-ary structure of $\gG$, all orbits $O^f_{\tuple v}, O^g_{\tuple v}$ are pp-definable from $\sA$, and hence $\psi$ is equivalent to a pp-formula over $\sA$. Since $\gG$ is $k$-neoliberal, and since $\sA$ has bounded strict width, we can proceed as in the proof of~\Cref{cor:noequality} and use~\Cref{lemma:l+1=rel,lemma:l+2=rel} to show that $\psi^{\sA}$ contains an injective mapping $h$. By the construction and by the $k$-homogeneity of $\gG$, this mapping satisfies our assumptions.
\end{proof}

The following observation states a few properties of implications and their compositions which will be used later.

\begin{observation}\label{observation:impl_properties}
    Let $\sA$ be a relational structure, let $k\geq 2$, let $C\subseteq A^k$, let $\phi_1$ be a $(C,\tuple u^1,C,\tuple v^1)$-implication in $\sA$, and let $\phi_2$ be a $(C,\tuple u^2,C,\tuple v^2)$-implication in $\sA$. Let $p_1$ be the number of variables of $\phi_1$, let $p_2$ be the number of variables of $\phi_2$, and let $p$ be the number of variables of $\phi:=\phi_1\circ \phi_2$. Then all of the following hold.

    \begin{enumerate}
        \item $p\geq\max(p_1,p_2)$.
        \item $p=p_1=p_2$ if, and only if, $\scope(\tuple u^1)\cap\scope(\tuple v^2) = \scope(\tuple u^1)\cap\scope(\tuple v^1)=\scope(\tuple u^1)\cap\scope(\tuple u^2)\cap\scope(\tuple v^2)$.
        \item Suppose that $\phi_1=\phi_2$ and $p=p_1$. Then for every $i\in[k]$, it holds that if $v^1_i$ is contained in the intersection in (2), then so is $u^1_i$.
    \end{enumerate}
\end{observation}

\begin{proof}
    Let us rename the variables of $\phi_2$ as in~\Cref{def:composition} so that $\tuple v^1=\tuple u^2$, and $\phi_1$ and $\phi_2$ do not share any further variables.
    
    For (1), observe that the number of variables of a $(C,\tuple u,C,\tuple v)$-pre-implication is equal to $2k-|\scope(\tuple u)\cap \scope(\tuple v)|$. By~\Cref{lemma:kimpl}, $\phi$ is a $(C,\tuple u_1,C,\tuple v_2)$-pre-implication, whence    $p=2k-|\scope(\tuple u^1)\cap \scope(\tuple v^2)|$, and since $\scope(\tuple u^1)\cap \scope(\tuple v^2)$ is contained both in $\scope(\tuple u^1)\cap \scope(\tuple v^1)$ and in $\scope(\tuple u^2)\cap \scope(\tuple v^2)$, (1) follows by applying the same reasoning to $p_1$ and $p_2$.

    For (2), observe that by the previous paragraph, $p=p_1$ if, and only if, $\scope(\tuple u^1)\cap\scope(\tuple v^2) = \scope(\tuple u^1)\cap\scope(\tuple v^1)$. Similarly, $p=p_2$ if, and only if, $\scope(\tuple u^1)\cap\scope(\tuple v^2) = \scope(\tuple u^2)\cap\scope(\tuple v^2)$, and (2) follows by the fact that $\tuple v^1=\tuple u^2$.
    
    For (3), suppose that $v^1_i\in \scope(\tuple u^1)$. It follows by (2) that $v^1_i=u^2_i\in\scope(\tuple v^2)$, and since $\phi_2$ was obtained by renaming variables of $\phi_1$, and in particular, $\tuple u^2$ was obtained by renaming $\tuple v^1$, it follows that $u^1_i\in\scope(\tuple v^1)$.
\end{proof}

\subsection{Digraphs of implications}

We reformulate the notion of digraph of implications from~\cite{Wrona:2020b} and prove a few auxiliary statements about these digraphs.

\begin{definition}\label{def:impl}
Let $k\geq 3$, let $\sA$ be a first-order expansion of a $k$-neoliberal relational structure $\sB$. Let $\emptyset\neq C\subseteq A^k$, and let $\phi$ be $(C,\tuple u,C,\tuple v)$-implication in $\sA$ such that $\proj_{\tuple u}(\phi^{\sA})=\proj_{\tuple v}(\phi^{\sA})=:E$. Let $\mathbf{Vert}(E)$ be the set of all orbits under $\Aut(\sB)$ contained in $E$.
Let $\mathcal{B}_{\phi}\subseteq \mathbf{Vert}(E)\times \mathbf{Vert}(E)$ be the directed graph such that $\mathcal{B}_{\phi}$ contains an arc $(O,P)\in \mathbf{Vert}(E)\times \mathbf{Vert}(E)$ if $\phi^{\sA}$ contains an $OP$-mapping.

We say that $S \subseteq \mathbf{Vert}(E)$ is a \emph{strongly connected component} if it is a maximal set with respect to inclusion such that for all (not necessary distinct) vertices $O,P\in S$, there exists a path in $\mathcal{B}_{\phi}$ connecting $O$ and $P$. We say that $S$ is a \emph{sink} in $\mathcal B_{\phi}$ if every arc originating in $S$ ends in $S$; $S$ is a \emph{source} in $\mathcal B_{\phi}$ if every arc finishing in $S$ originates in $S$.
\end{definition}

Note that the digraph $\mathcal B_\phi$ can be defined also for relational structures which do not satisfy the assumptions on $\sA$ from~\Cref{def:impl}; however, these assumptions are needed in the proof of~\Cref{lemma:kcomplete} so we chose to include them already in~\Cref{def:impl}. Note also that $\mathcal{B}_\phi$ can contain vertices which are not contained in any strongly connected component.

\begin{observation}\label{observation:sinksource}
Let $\phi$ be as in~\Cref{def:impl}.  
Then there exist strongly connected components $S_1\subseteq\mathbf{Vert}(C),S_2\subseteq\mathbf{Vert}(E\backslash C)$ in $\mathcal B_\phi$ such that $S_1$ is a sink in $\mathcal{B}_{\phi}$, and $S_2$ is a source in $\mathcal{B}_{\phi}$.
Moreover, since $\proj_{\tuple u}(\phi^{\sA})=\proj_{\tuple v}(\phi^{\sA})=E$, any vertex $O\in \mathbf{Vert}(E)$ has an outgoing and an incoming arc, i.e., the digraph $\mathcal{B}_{\phi}$ is \emph{smooth}.
\end{observation}

\begin{proof}
    The second part of the lemma is immediate.
    To prove the first part, observe that since $\phi$ is a $(C,C)$-implication in $\sA$, it follows that $\mathbf{Vert}(C)$ is a sink in $\mathcal{B}_{\phi}$. Using the oligomorphicity of $\gG$ and the definition of a $(C,C)$-implication, we get that the induced subgraph of $\mathcal B_\phi$ on $\mathbf{Vert}(C)$ is finite and smooth. Hence, there exists a strongly connected component $S_1\subseteq \mathbf{Vert}(C)$ in $\mathcal{B}_\phi$ which is a sink in the induced subgraph, and hence also in $\mathcal B_\phi$.
    Similarly, one sees that $\mathbf{Vert}(E\backslash C)$ is a source in $\mathcal{B}_{\phi}$, and it contains a strongly connected component $S_2$ which is itself a source in $\mathcal B_{\phi}$.
\end{proof}

Let $\phi$ be as in~\Cref{def:impl}, and set $I_\phi:=\{i\in[k]\mid u_i\in\scope(\tuple v)\}$. If the number of variables of $\phi$ is equal to the number of variables of $\phi\circ \phi$, then item (3) in~\Cref{observation:impl_properties} yields that $I_\phi=\{i\in[k]\mid v_i\in\scope(\tuple u)\}$. 

\begin{definition}\label{def:kcomplimpl}
Let $\phi$ be as in~\Cref{def:impl}. We say that $\phi$ is \emph{complete} if the number of variables of $\phi\circ \phi$ is equal to the number of variables of $\phi$, $u_i=v_i$ for every $i\in I_\phi$, and each strongly connected component of $\mathcal{B}_{\phi}$ contains all possible arcs including loops.
\end{definition}

The following is a modification of Lemma 36 in \cite{Wrona:2020b}:

\begin{lemma}\label{lemma:kcomplete}
Let $\phi$ be as in~\Cref{def:impl}.
Then there exists a complete injective $(C,C)$-implication in $\sA$.
\end{lemma}

\begin{proof}
We will construct the desired complete implication as a conjunction of a power of $\phi$ and $I_\ell$, where $\ell$ is the number of variables of the power of $\phi$. Note that for every $n\geq 1$, the number of variables of $\phi^{\circ n}$ is at most $2k$ by~\Cref{def:composition} and this number never decreases with increasing $n$ by item (1) in~\Cref{observation:impl_properties}. 
Hence, there is $n\geq 1$ such that the number of variables of $\phi^{\circ n}$ is the biggest among all choices of $n$. Let us denote the number of variables of $\phi^{\circ n}$ by $\ell$; it follows that the number of variables of $(\phi^{\circ n})^{\circ m}$ is equal to $\ell$ for every $m\geq 1$. Let us replace $\phi$ by $\phi^{\circ n}$. It follows from~\Cref{lemma:kimpl} that for every $m\geq 1$, $\phi^{\circ m}\wedge I_\ell$ is an injective $(C,C)$-implication. Now, it follows by item (3) in \Cref{observation:impl_properties} that there exists a unique bijection $\sigma\colon I_\phi\rightarrow I_\phi$ such that $u_i=v_{\sigma(i)}$ for every $i\in I_\phi$.
Replacing $\phi$ with a power of $\phi$ again, we can assume that $\sigma$ is the identity.

Now, we can find $m\geq 1$ such that the number of strongly connected components  of $\phi^{\circ m}$ is maximal among all possible choices of $m$. It follows that composing $\phi^{\circ m}$ with itself arbitrarily many times does not disconnect any vertices from $\mathbf{Vert}(C)$ which are contained in the same strongly connected component of $\mathcal{B}_{\phi^{\circ m}}$; we replace $\phi$ by $\phi^{\circ m}$. Taking another power of $\phi$ and replacing $\phi$ again, we can assume that every strongly connected component of $\mathcal B_\phi$ contains all loops. Now, setting $p$ to be the number of vertices of $\mathcal B_\phi$, we have that, replacing $\phi$ with $\phi^{\circ p}$, every strongly connected component of $\mathcal B_\phi$ contains all arcs, whence $\phi\wedge I_\ell$ is a complete injective $(C,C)$-implication.
\end{proof}

\subsection{Proof of~\Cref{thm:implsimple}}

\implsimple*

\begin{proof}
Let $\sA$ be a first-order expansion of $\sB$ with bounded strict width. Striving for contradiction, 
suppose that $\sA$ is implicationally hard on injective instances. Then the injective implication graph $\mathcal{G}_{\sA}^{\injinstances}$ contains a directed cycle $(D_1,C_1),\ldots,(D_{n-1},C_{n-1}),(D_n,C_n)=(D_1,C_1)$. This means that for all $i\in[n-1]$, there exists an injective $(C_i,\tuple u^i,C_{i+1},\tuple u^{i+1})$-implication $\phi_i$ in $\sA$ with $\proj_{\tuple u^i}(\phi_i^{\sA})=D_i$, and $\proj_{\tuple u^{i+1}}(\phi_i^{\sA})=D_{i+1}$.

Let us define $\phi:=((\phi_1\circ \phi_2)\circ\ldots\circ \phi_{n-1})$. Restricting $\phi^{\sA}$ to injective mappings, we obtain an injective $(C_1,\tuple u^1,C_1,\tuple u^n)$-implication by~\Cref{lemma:kimpl}. \Cref{lemma:kcomplete} asserts us that there exists a complete injective $(C_1,C_1)$-implication $\psi$ in $\sA$.

\Cref{observation:sinksource} yields that there exist $C\subseteq C_1$, and $D\subseteq D_1\backslash C_1$ such that $\mathbf{Vert}(C)$ is a strongly connected component which is a sink in $\mathcal{B}_{\psi}$, and $\mathbf{Vert}(D)$ is a strongly connected component which is a source in $\mathcal{B}_{\psi}$. 
Observe that since $\sA$ is a first-order expansion of $\sB$ and since $\psi$ is complete, $C$ is pp-definable from $\sA$. Indeed, for any fixed orbit $O\subseteq C$ under $\Aut(\sB)$, $C$ is equal to the set of all orbits $P\subseteq C_1$ such that $\psi^{\sA}$ contains an $O P$-mapping. We can observe in a similar way that $D$ is pp-definable from $\sA$ as well. 
Moreover, $\psi$ is easily seen to be a complete $(C,C)$-implication in $\sA$.

Since $\sB$ has finite duality, there exists a number $d\geq 3$ such that for every finite structure $\sX$ in the signature of $\sB$, it holds that if every substructure of $\sX$ of size at most $d-2$ maps homomorphically to $\sB$, then so does $\sX$. 
Set $\rho:=\psi^{\circ d}$. Let $V$ be the set of variables of $\rho$. It follows from~\Cref{lemma:kimpl} that $\rho$ is
a $(C,\tuple u,C,\tuple v)$-implication in $\sA$ for some $\tuple u,\tuple v$.
We are going to prove the following claim using the finite duality of $\sB$ and the completeness of $\psi$.

\begin{claim}\label{claim:finduality}
    Every $f\in A^{V}$ with $f(\tuple u)\in C$ and $f(\tuple v)\in C$ is an element of $\rho^{\sA}$. The same holds for any  $f\in A^{V}$ with $f(\tuple u) \in D$ and $f(\tuple v)\in D$.
\end{claim}

To prove~\Cref{claim:finduality}, let $f\in A^V$ be as in the statement of the claim. 
Up to renaming variables, we can assume that $\psi$ is a $(C,\tuple u,C,\tuple v)$-implication. 
Completeness of $\psi$ implies that $I_\psi=\{i\in[k]\mid u_i\in\scope(\tuple v)\}=\{i\in[k]\mid v_i\in\scope(\tuple u)\}$, and that $u_i=v_i$ for every $i\in I_\psi$. Let $\tuple w^0,\ldots,\tuple w^d$ be $k$-tuples of variables such that $w^j_i=u_i=v_i$ for all $i\in I_\psi$ and all $j\in\{0,\ldots,d\}$, and disjoint otherwise. We can moreover assume that $\tuple w^0=\tuple u$, $\tuple w^d=\tuple v$. 
For every $j\in\{0,\ldots,d-1\}$, let $\psi_{j+1}$ be the $(C,\tuple w^j,C,\tuple w^{j+1})$-implication obtained from $\psi$ by renaming $\tuple u$ to $\tuple w^j$ and $\tuple v$ to $\tuple w^{j+1}$. 
It follows that $\rho$ is equivalent to the formula obtained from $\psi_1\wedge\dots\wedge \psi_d$ by existentially quantifying all variables that are not contained in $\scope(\tuple u)\cup\scope(\tuple v)$. 
In order to proof~\Cref{claim:finduality}, it is therefore enough to show that $f$ can be extended to a mapping $h\in (\psi_1\wedge\dots\wedge\psi_d)^{\sA}$.

Now, for all $p,q\in[k]$, we identify the variables $w^0_{q}=u_q$ and $w^{d}_{p}=v_p$ if $f(v_{p}) = f(u_{q})$. Observe that $f|_{\scope(\tuple u)\cap \scope(\tuple v)}$ is injective since $f(\tuple u)\in C\subseteq I^A_k$, and hence this identification does not force any variables from $\scope(\tuple u)\cap\scope (\tuple v)$ to be equal. Moreover, since $d\geq 2$, this identification does not identify any variables from the tuples $\tuple w^{d-1}$ and $\tuple w^d$. Let us define $f_0:=f|_{\scope(\tuple w^0)\cup \scope(\tuple w^d)}$. 
It is enough to show that $f_0$ can be extended to a mapping $h\in (\psi_1\wedge \dots\wedge \psi_d)^{\sA}$.

Let $O$ be the orbit of $f(\tuple u)$ under $\Aut(\sB)$, and let $P$ be the orbit of $f(\tuple v)$. Let $f_1\in \psi_1^{\sA}$ be an injective $O P$-mapping, and for every $j\in\{2,\ldots,d\}$, let $f_j\in \psi_j^{\sA}$ be an injective $P P$-mapping. 
Note that such $f_j$ exists for every $j\in[d]$ since $\psi_j$ is complete and since $\mathbf{Vert}(C)$ and $\mathbf{Vert}(D)$ are strongly connected components in $\mathcal{B}_{\psi_j}$. 

Let $\tau$ be the signature of $\sB$. Let $X:=\scope(\tuple w^0)\cup\cdots\cup \scope(\tuple w^d)$, and let us define a $\tau$-structure $\sX$ on $X$ as follows. Recall that the relations from $\tau$ correspond to the orbits of injective $k$-tuples under $\Aut(\sB)$. For every relation $R\in\tau$, we define $R^{\sX}$ to be the set of all tuples $\tuple w\in I^X_k$ such that there exists $j\in \{0,\ldots,d\}$ such that $\scope(\tuple w)\subseteq \scope(\tuple w^j)\cup\scope(\tuple w^{j+1})$ and $f_j(\tuple w)\in R$; here and in the following, the addition $+$ on indices is understood modulo $d+1$. We will show that $\sX$ has a homomorphism $h$ to $\sB$. If this is the case, it follows by the construction and by the $k$-homogeneity of $\Aut(\sB)$ that $h\in (\psi_1\wedge\dots\wedge\psi_d)^{\sA}$. Moreover, we can assume that $h|_{\scope(\tuple w^0)\cup\scope(\tuple w^d)}=f_0$ as desired.

Let now $Y\subseteq X$ be of size at most $d-2$. 
Then for cardinality reasons there exists $j\in[d-1]$ such that $Y\subseteq \scope(\tuple w^0)\cup\cdots\cup \scope(\tuple w^{j-1})\cup \scope(\tuple w^{j+1})\cup\cdots\cup \scope(\tuple w^{d})$. 
Observe that $f_{j+1}$ is a homomorphism from the induced substructure of $\sX$ on $\scope(\tuple w^{j+1})\cup \scope(\tuple w^{j+2})$ to $\sA$. 
Since the orbits of $f_{j+1}(\tuple w^{j+2})$ and $f_{j+2}(\tuple w^{j+2})$ agree by definition, by composing $f_{j+2}$ with an element of $\Aut(\sB)$ we can assume that $f_{j+1}(\tuple w^{j+2})=f_{j+2}(\tuple w^{j+2})$. 
We can proceed inductively and extend $f_{j+1}$ 
to $\scope(\tuple w^0)\cup\cdots\cup \scope(\tuple w^{j-1})\cup \scope(\tuple w^{j+1})\cup\cdots\cup \scope(\tuple w^{d})$ such that it is a homomorphism from the induced substructure of $\sX$ on this set to $\sA$. It follows that the substructure of $\sX$ induced on $Y$ maps homomorphically to $\sA$. Finite duality of $\sB$ yields that $\sX$ has a homomorphism to $\sA$ as desired and \Cref{claim:finduality} follows.
\medskip

Assume without loss of generality that $1\notin I_\rho$, and identify $u_i$ with $v_i$ for every $i\neq 1$. Note that this is possible by item (3) in \Cref{observation:impl_properties}, and since $\rho$ is easily seen to be complete.
Set $W:=\scope(\tuple u)\cup\scope(\tuple v)$, and let $\rho'$ be the formula arising from $\rho$ by this identification.
We will argue that $\rho'$ is critical in $\sA$ over $(C,D,\tuple u,\tuple v)$.
To this end, let us first show that $\rho'$ is a $(C,\tuple u,C,\tuple v)$-implication in $\sA$. Observe that for every orbit $O\subseteq C$ under $\Aut(\sB)$, $(\rho')^{\sA}$ contains an injective $OO$-mapping. Indeed, there exists an injective $g\in A^W$ such that $g(\tuple u)\in O$ and $g(\tuple v)\in O$; this easily follows by the $(k-1)$-transitivity of $\Aut(\sB)$ and by the fact that it has no $k$-algebraicity. Forgetting the identification of variables, we can understand $g$ as an element of $A^V$, and~\Cref{claim:finduality} yields that $g\in\rho^{\sA}$, whence $g\in (\rho')^{\sA}$. 
Now, it immediately follows that $\rho'$ satisfies the items (1)-(3) and (5) from \Cref{def:implication}. Moreover, the satisfaction of item (4) follows immediately from the fact that $\rho$ is a $(C,\tuple u,C,\tuple v)$-implication in $\sA$.

It remains to verify the last three items of \Cref{def:kcritical}. Observe similarly as above that for any orbit $O\subseteq D$ under $\Aut(\sB)$, $(\rho')^{\sA}$ contains an $OO$-mapping, which immediately yields that $D$ is contained both in $\proj_{\tuple u}((\rho')^{\sA})$ and in $\proj_{\tuple v}((\rho')^{\sA})$, it also yields that for every $\tuple a\in D$, there exists
$f\in(\rho')^{\sA}$ such that $f(\tuple u)\in D$ and $f(\tuple v)= \tuple a$. Hence, $\rho'$ is indeed critical in $\sA$ over $(C,D,\tuple u,\tuple v)$, contradicting \Cref{lemma:kcritical}.
\end{proof}

\kmain*

\begin{proof}
    Let $\sA$ be a first-order expansion of $\sB$ with bounded strict width.
    By~\Cref{cor:libcores-purelyinj}, it is enough to prove that $\Csp_{\injinstances}(\sA)$ has relational width $(k,\max(k+1,b_{\sB}))$.
    \Cref{thm:implsimple} yields that $\sA$ is implicationally simple on injective instances and the result follows from~\Cref{prop:implsimple}.
\end{proof}

\bibliographystyle{plainurl}
\bibliography{main}
\end{document}